\documentclass[a4paper,12pt]{article}

\usepackage{amsmath,amsthm,amssymb,multicol,epsf}
\usepackage{color}
\usepackage{graphicx}
\usepackage{amsmath}
\usepackage{amsfonts}
\usepackage{amssymb}
\usepackage{amsthm}
\usepackage{epstopdf}
\usepackage{algorithmic}
\usepackage{setspace}
\usepackage[cm]{fullpage}
\usepackage{multirow}

\newtheorem{defi}{Definition}[section]
\newtheorem{theorem}{Theorem}[section]
\newtheorem{lem}{Lemma}[section]

\newtheorem{proposition}{Proposition}[section]

\theoremstyle{remark}
\newtheorem{algorithm}{Algorithm}

\newcommand{\mbf}[1]{\mbox{\boldmath$#1$}}

\usepackage{xcolor}

\numberwithin{equation}{section}

\title{An improved return-mapping scheme for nonsmooth yield surfaces: PART I - the Haigh-Westergaard coordinates}
\author{S. Sysala$^1$, M. Cermak$^{1,2}$, T. Koudelka$^3$, J. Kruis$^3$, J.
Zeman$^{2,3}$, R. Blaheta$^1$\\ \\ $^1$Institute of Geonics CAS, Ostrava, Czech
Republic\\ $^2$V\v SB--Technical University of Ostrava, Ostrava, Czech Republic\\
$^3$Czech Technical University in Prague, Prague, Czech Republic}

\begin{document}

\maketitle

\begin{abstract}
The paper is devoted to the numerical solution of elastoplastic constitutive initial value problems. An
improved form of the implicit return-mapping scheme for nonsmooth yield surfaces is proposed that systematically builds on a
subdifferential formulation of the flow rule. The main advantage of this approach is that the treatment
of singular points, such as apices or edges at which the flow direction is multivalued involves only a uniquely defined set of non-linear
equations, similarly to smooth yield surfaces. This paper (PART I) is focused on isotropic models containing: $a)$ yield surfaces with one or two apices (singular points) laying on the hydrostatic axis; $b)$ plastic pseudo-potentials that are independent of the Lode angle; $c)$ nonlinear isotropic hardening (optionally). It is shown that for some models the improved integration scheme also enables to a priori decide about a type of the return and investigate existence,  uniqueness and semismoothness of discretized constitutive operators in implicit form. Further,  the semismooth Newton method is introduced to solve incremental boundary-value problems. The paper also contains numerical examples related to slope stability with available Matlab implementation.
\end{abstract}



\section{Introduction}
\label{sec_int}

The paper is devoted to the numerical solution of small-strain
quasi-static elastoplastic problems. Such a problem consists of the constitutive
initial value problem (CIVP) and the balance equation representing the principle
of virtual work. A broadly exploited and universal numerical/computational
concept includes the following steps: 

\begin{itemize}

\item[$(a)$] time-discretization of CIVP leading to an incremental constitutive
problem; 
\item[$(b)$] derivation of the constitutive and consistent tangent operators;
\item[$(c)$] substitution of the constitutive (stress-strain) operator into the
balance equation leading to the incremental boundary value problem in terms of displacements; 
\item[$(d)$] finite element
discretization and derivation of a system of nonlinear equations; 
\item[$(e)$] solving
the system using a nonsmooth variant of the Newton method. 
\end{itemize}

CIVP satisfies thermodynamical laws and usually involves internal variables such as plastic strains
or hardening parameters. Several integration schemes for numerical solution of CIVP were suggested. For their overview, we refer, e.g., \cite{AT15, NPO08, SH98, SHVS14} and references introduced therein. If the implicit or trapezoidal Euler method is
used then the incremental constitutive problem is solved by the elastic
predictor/plastic corrector method. The plastic correction leads to the
return-mapping scheme. We distinguish, e.g., implicit, trapezoidal or midpoint
return-mappings depending on a chosen time-discretization \cite[Chapter 7]{NPO08}. 


In this paper, we assume that the plastic flow direction is generated by the plastic potential, $g$. If $g$ is smooth then the corresponding plastic flow direction is uniquely determined by the derivative of $g$ and consequently, the plastic flow rule reads as follows, e.g.~\cite[Chapter 8]{NPO08}:
\begin{equation}
\dot{\mbf{\varepsilon}}^p=\dot\lambda\frac{\partial g(\mbf\sigma,A)}{\partial \mbf\sigma},\quad g=g(\mbf\sigma,A).
\label{eqn_flow_rule}
\end{equation}
Here, 
$\dot{\mbf{\varepsilon}}^p$, $\dot\lambda$, $\mbf\sigma$,  and $A$
denotes the plastic strain rate, the  plastic multiplier
rate, the stress tensor and the hardening thermodynamical forces, respectively. The corresponding return-mapping scheme is relatively straightforward and leads to solving
a system of nonlinear equations. A difficulty arises when $g$ is {\it nonsmooth}. Mostly, it happens if the yield surface contains singular points, such as apices or edges. Then the function $g$ is rather pseudo-potential than potential and its derivative need not exist everywhere. In such a case, the rule (\ref{eqn_flow_rule}) is usually completed by some additive formulas depending on particular cases of $g$ and $\mbf\sigma$ in an ad-hoc manner. For example, the
implementation of the Mohr-Coulomb model reported in~\cite[Chapter 6, 8]{NPO08} employs one, two, or six
plastic multipliers $\lambda$, depending on the location of $\mbf\sigma$ on the yield
surface. Since the stress tensor $\mbf\sigma$ is unknown in CIVP one must
blindly guess its right location. Moreover, for each tested location, one must
usually solve an auxilliary system of nonlinear equations whose solvability is
not guaranteed in general. \emph{These facts are evident drawbacks of the
current return-mapping schemes.} 

In associative plasticity, it is well-known that the plastic flow rule
(\ref{eqn_flow_rule}) together with a hardening law and loading/unloading
conditions can be equivalently replaced by the principle of maximum plastic
dissipation within the constitutive model. This alternative formulation of CIVP
does not require special treatment for nonsmooth $g$ and enables
to solve CIVP by techniques based on mathematical programming \cite{CM91, FZ88,
RM91}. In particular, if the implicit or trapezoidal Euler method is used then
the incremental constitutive problem can be interpreted by a certain kind of the
closest-point projection \cite{AP02, PA02, Sy14}. For some nonassociative
models, CIVP can be re-formulated using a theory of bipotentials that leads to
new numerical schemes \cite{S95, HFS02, Z09}. These alternative
definitions of the flow rule enable a
variational re-formulation of the initial boundary value elastoplastic problem.
Consequently, solvability of this problem can be investigated (see, e.g.,
\cite{HR99,MR15}). Therefore, the corresponding numerical techniques are usually
also correct from the mathematical point of view. On the other hand,
such a numerical treatment is not so universal and its implementation is more
involved/too complex in comparison with standard procedures of computational
inelasticity.

The approach pursued in
this paper builds on the subdifferental formulation of the plastic flow rule,
e.g. \cite[Section 6.3.9]{NPO08},
\begin{equation}
\dot{\mbf{\varepsilon}}^p\in\dot\lambda\partial_{\sigma} g(\mbf\sigma,A)
\label{inclusion_flow_rule}
\end{equation}
for nonsmooth $g$. Here, $\partial_\sigma g(\mbf\sigma,A)$ denotes the subdifferential of $g$ at
$(\mbf\sigma,A)$ with respect to the stress variable. If $g$ is convex at least
in vicinity of the yield surface then this definition is justified, e.g., by
\cite[Corollary 23.7.1]{Ro70} and is valid even when $g$
is not smooth at $\mbf\sigma$. On the first sight, it seems
that (\ref{inclusion_flow_rule}) is not convenient for numerical treatment due
to the presence of the multivalued flow direction. The main goal of this paper
is to show that the \emph{opposite is true}, by demonstrating that the
implicit return-mapping scheme based on (\ref{inclusion_flow_rule}) leads to
solving a just one system of nonlinear equations regardless whether the
unknown stress tensor lies on the smooth portion of the yield surface or not at least for a
wide class of models with nonsmooth plastic pseudo-potentials. Using this technique, we
eliminate the blind guessing and thus considerably
simplify the solution scheme. Moreover, the new technique enables to investigate some useful properties of the constitutive operator, like uniqueness or semismoothness, that are not obvious for the current technique.

\subsection{Basic idea}

First of all, we illustrate the new technique on a simple 2D projective problem
that mimics the structure of an incremental elastoplastic constitutive
problem. Consider the convex set $$B:=\{\mbf w=(w_1,w_2)\in\mathbb R^2\ |\;
f(\mbf w)\leq0\},\quad f(\mbf w):=w_1+|w_2|-1,$$ and define the projection $\mbf
w^*\in B$ of a point $\mbf z=(z_1,z_2)\in\mathbb R^2$ as follows: $$\|\mbf
z-\mbf w^*\|^2=\min_{\mbf w\in B}\|\mbf z-\mbf w\|^2,\quad \|\mbf
z\|^2:=z_1^2+z_2^2.$$ The scheme of the projection is depicted in Figure
\ref{fig_projection}.

\begin{figure}[htbp]
        \begin{center}
          \includegraphics[width=0.3\textwidth]{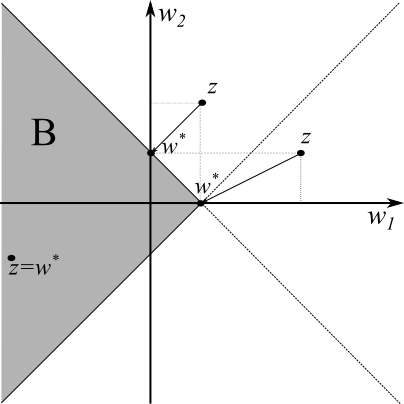}
        \end{center}
\caption{Scheme of the projection.}
\label{fig_projection}
\end{figure}

Clearly, the function $f$ is convex in $\mathbb R^2$, nondifferentiable at $\mbf w=(w_1,0)$ and 
\begin{equation}
\nabla f(\mbf w)=\left(
1,\ \frac{w_2}{|w_2|}
\right)^T\in\mathbb R^2\quad\forall \mbf w=(w_1,w_2),\; w_2\neq0.
\label{motivation0}
\end{equation}
If $\mbf z\in B$ then $\mbf w^*=\Pi_B(\mbf z)=\mbf z$. Conversely, if $\mbf
z\not\in B$ it follows from the Karush-Kuhn-Tucker conditions and
(\ref{motivation0}) that the projective problem can be written 
as follows: \textit{find $\mbf w^*=(w_1^*,w_2^*)^T\in\mathbb R^2$ and
the Lagrange multiplier $\lambda>0$:
\begin{equation}
z_1-w_1^*=\lambda,\quad z_2-w_2^*\in\lambda\partial|w_2^*|,\quad w_1^*+|w_2^*|-1=0,
\label{motivationb2}
\end{equation}
where}
$$\partial|w_2^*|=\left\{
\begin{array}{l l}
\{w_2^*/|w_2^*|\},& w_2^*\neq0,\\[0pt]
[-1,1], & w_2^*=0,
\end{array}
\right.$$
To find a solution to (\ref{motivationb2}),  it is crucial to rewrite the
inclusion (\ref{motivationb2})$_2$ as an equation. Observe that
$$z_2-w_2^*\in\lambda\partial|w_2^*|\quad\mbox{if and only if}\quad w_2^*=(|z_2|-\lambda)^+\frac{z_2}{|z_2|},$$ where $(\cdot)^+$ denotes a positive part of a function. This crucial transformation will be derived in detail in Section \ref{sec_DP} on an analogous elastoplastic example.
Thus (\ref{motivationb2}) leads to the following system of equations:
\begin{equation}
w_1^*=z_1-\lambda,\quad w_2^*=(|z_2|-\lambda)^+\frac{z_2}{|z_2|},\quad w_1^*+|w_2^*|-1=0,
\label{motivationb3}
\end{equation}
Since (\ref{motivationb3})$_2$ implies $|w_2^*|=(|z_2|-\lambda)^+$, the system
of three nonlinear equations reduces to a single one
$$z_1-\lambda+(|z_2|-\lambda)^+-1=0.$$
Consequently, $\lambda$ can be found in the closed form as
$$\lambda=z_1-1+\frac{1}{2}(-z_1+|z_2|+1)^+=\left\{
\begin{array}{c l}
\frac{1}{2}(z_1+|z_2|-1),& z_1-|z_2|-1\leq0,\\
z_1-1,& z_1-|z_2|-1\geq0
\end{array}
\right.$$
from which one can easily compute $\mbf w^*=(w_1^*,w_2^*)^T$ by
(\ref{motivationb3})$_1$ and (\ref{motivationb3})$_2$.

\subsection{Content of the the paper}

The presented idea is systematically extended on some elastoplastic models. This paper, PART I, is focused on isotropic models containing: $a)$ yield surfaces with one or two apices (singular points) laying on the hydrostatic axis; $b)$ plastic pseudo-potentials that are independent of the Lode angle; $c)$ nonlinear isotropic hardening (optionally). Such models are usually formulated by
the Haigh-Westergaard coordinates. Further, the implicit Euler discretization of CIVP is considered and thus two types of return on the yield surface within the plastic correction are distinguished: $(i)$ return to the smooth portion of the yield surface; $(ii)$ return to the apex (apices).

The paper is organized as follows. Section \ref{sec_preliminaries} contains some preliminaries related to invariants of the stress tensor and semismooth functions. Section
\ref{sec_DP} is devoted to the Drucker-Prager model including the nonlinear isotropic hardening. Although the plastic corrector cannot be found in closed form, the new technique enables to   a priori decide about the return type and prove existence, uniqueness and semismoothness of the implicit constitutive operator. The consistent tangent operator is also introduced. In Section \ref{sec_JG}, we derive similar results for the perfect plastic part of the Jir\'asek-Grassl model \cite{GJ06}.
In Section \ref{sec_general}, the new technique is extended on an abstract model written by the Haigh-Westergaard coordinates. In particular, within the plastic correction, we formulate a unique system of nonlinear equations which is common for the both type of the return. It can lead to a more correct and/or simpler solution scheme in comparison with the current technique. Section \ref{sec_realization} is devoted to numerical realization of the incremental boundary value elastoplastic problem using the semismooth Newton method. In Section \ref{sec_experiments},  illustrative numerical examples related to a slope stability benchmark are considered. Here, limit load is analyzed by an incremental method depending on a mesh type and mesh density for the Drucker-Prager and Jir\'asek-Grassl models.

Within this paper, second order tensors, matrices and vectors are denoted by
bold letters. As usual, small letters are used for vectors and capitals for
matrices (see Section \ref{sec_realization}). Further, the fourth order tensors
are denoted by capital blackboard letters, e.g., $\mathbb D_e$ or $\mathbb
I_{dev}$. The symbol $\otimes$ means the tensor product \cite{NPO08, G81}. We also use the following
notation: $\mathbb R_+:=\{z\in\mathbb R;\; z\geq0\}$ and $\mathbb R^{3\times
3}_{sym}$ for the space of symmetric, second order tensors.


\section{Preliminaries}
\label{sec_preliminaries}

\subsection{Invariants of a stress tensor and their derivatives}
\label{sec_invariants}

Consider a stress tensor $\mbf\sigma\in\mathbb R^{3\times 3}_{sym}$ and its splitting into the volumetric and deviatoric parts:
$$\mbf\sigma=p\mbf I+\mbf s,\quad p:=p(\mbf\sigma)=\frac{1}{3}\mbf I:\mbf\sigma,\;\;\mbf s:= s(\mbf\sigma)=\mathbb I_{dev}:\mbf \sigma=\mbf\sigma-\frac{1}{3}(\mbf I:\mbf\sigma)\mbf I.$$
Here, $\mbf I$, $\mathbb I_{dev}$, $p$ and $\mbf s$ denote the identity second order tensor, the fourth order deviatoric projection tensor, the hydrostatic pressure, and the deviatoric stress, respectively.  {\it The Haigh-Westergaard coordinates} are created by the invariants $p$, $\varrho$ and $\theta$, where
$$\varrho:=\varrho(\mbf\sigma)=\sqrt{\mbf s:\mbf s}=\|\mbf s\|,$$
\begin{equation*}
\theta:=\theta(\mbf\sigma)=\frac{1}{3}\arccos\left(\frac{3\sqrt{3}}{2}\frac{J_3}{J_2^{3/2}}\right),\;\; J_2(\mbf s)=\frac{1}{2}\mbf s:\mbf s=\frac{1}{2}\varrho^2,\;\; J_3(\mbf s)=\frac{1}{3}\mbf s^3:\mbf I,\;\; \mbf s:=\mbf s(\mbf\sigma).
\label{theta}
\end{equation*}
Clearly, $\varrho\geq0$ and $\theta\in[0,\pi/3]$. Since $\theta$ is not well-defined when $\varrho=0$, the Lode angle is included in elastoplastic models indirectly. Usually, it is considered another invariant in the form
\begin{equation}
\tilde\varrho:=\tilde\varrho(\mbf\sigma)=\varrho\tilde r(\cos\theta),
\label{invariant_abstract}
\end{equation}
where $\tilde r(\cdot)$ is assumed to be a smooth function such that $\tilde\varrho(\cdot)$ is at least continuous. However, as we will see below, it is more convenient to assume {\it strong semismoothness} of $\tilde\varrho(\cdot)$. As a particular case, it will be considered the invariant
\begin{equation}
\varrho_e:=\varrho_e(\mbf\sigma)=\varrho r_e(\cos\theta),
\label{invariant_e}
\end{equation}
where
\begin{equation}
r_e(\cos\theta)=\frac{4(1-e^2)\cos^2\theta+(2e-1)^2}{2(1-e^2)\cos\theta+(2e-1)\sqrt{4(1-e^2)\cos^2\theta+5e^2-4e}}.
\label{r_e}
\end{equation}
The function $r_e$ was proposed in \cite{WW74} and contains the excentricity parameter $e\in[0.5,1]$.  It holds: $a)$ $r_e(\cos\theta(.))$ is a bounded and smooth function for any $\mbf\sigma\in\mathbb R^{3\times 3}_{sym}$, $\varrho(\mbf\sigma)>0$; $b)$ $\varrho_e(\mbf\sigma)=0$ when $\varrho=0$; $c)$ $\varrho_e=\varrho$, $r_e(\cos\theta)=1$ when $e=1$.

We will also use the following derivatives:
\begin{equation}
\frac{\partial p}{\partial\mbf\sigma}=\frac{\mbf I}{3},\quad
\frac{\partial \mbf s}{\partial\mbf\sigma}=\mathbb I_{dev},\quad
\mbf{n}(\mbf\sigma):=\frac{\partial \varrho}{\partial\mbf\sigma}=\frac{\mbf s}{\varrho},\quad
\frac{\partial\mbf{n}}{\partial\mbf\sigma}=\frac{1}{\varrho}\left(\mathbb I_{dev}-\mbf{n}\otimes\mbf{n}\right),
\label{inv_der1}
\end{equation}
\begin{equation}
\frac{\partial \theta}{\partial\mbf\sigma}=\frac{\sqrt{6}}{\varrho\sin(3\theta)}\left[(\mbf{n}\otimes\mbf{n}^3)\mbf I-\mathbb I_{dev}(\mbf{n}^2)\right],\quad \frac{\partial r_e}{\partial\mbf\sigma}=-r_e'(\cos\theta)\sin\theta\frac{\partial \theta}{\partial\mbf\sigma}.
\label{inv_der2}
\end{equation}
Notice that the derivatives of $\varrho$, $\mbf n$, $\theta$ and $r_e$ do not exist when $\varrho=0$. Further, $\theta$ is not differentiable when $\mbf\sigma$ satisfies either $\theta=0$ or $\theta=\pi/3$. On the other hand, $r_e$ has derivatives for such stresses \cite{WW74}.

For purposes of this paper, it is crucial to derive the subdifferential of $\varrho$ at $\mbf\sigma$ when $\varrho(\mbf\sigma)=0$:
\begin{eqnarray}
\partial\varrho(\mbf\sigma)&=&\{\hat{\mbf n}\in\mathbb R^{3\times3}_{sym}\ |\;\;\varrho(\mbf \tau)\geq\varrho(\mbf\sigma)+\hat{\mbf n}:(\mbf \tau-\mbf\sigma)\;\;\forall\mbf\tau\in \mathbb R^{3\times3}_{sym}\}\nonumber\\
&=&\{\hat{\mbf n}\in\mathbb R^{3\times3}_{sym}\ |\;\;\|\mbf s(\mbf \tau)\|\geq(\hat{\mbf n}:\mbf I)(p(\mbf \tau)-p(\mbf\sigma))+\hat{\mbf n}: \mbf s(\mbf \tau)\;\;\forall\mbf\tau\in \mathbb R^{3\times3}_{sym}\}\nonumber\\
&=&\{\hat{\mbf n}\in\mathbb R^{3\times3}_{sym}\ |\;\;\mbf I:\hat{\mbf n}=0,\;\;\| \mbf s(\mbf \tau)\|\geq\hat{\mbf n}:\mbf  s(\mbf \tau)\;\;\forall\mbf\tau\in \mathbb R^{3\times3}_{sym}\}\nonumber\\
&=&\{\hat{\mbf n}\in \mathbb R^{3\times3}_{sym}\ |\;\; \mbf I:\hat{\mbf n}=0,\;\;\|\hat{\mbf n}\|\leq1\}\quad\mbox{if }\varrho(\mbf\sigma)=0.
\label{rho_subgrad}
\end{eqnarray}
If $\varrho(\mbf\sigma)>0$ then $\partial\varrho(\mbf\sigma)=\{\mbf n(\mbf\sigma)\}$ by (\ref{inv_der1})$_3$.
It is readily seen that 
\begin{equation}
\mbf I:\hat{\mbf n}=0\quad\forall \hat{\mbf n}\in\partial\varrho(\mbf\sigma),
\label{rho_subgrad_prop}
\end{equation}
regardless $\varrho(\mbf\sigma)= 0$ or not. 

\subsection{Semismooth functions}

\label{sub_semi}

Semismoothness was originally introduced by Mifflin \cite{Mi77} for functionals. Qi and J. Sun \cite{QS93} extended the definition of semismoothness to vector-valued functions to investigate the superlinear convergence of the Newton method. We introduce a definition of {\it strongly semismooth} functions \cite{Su01,Go04,PSS03}. To this end, consider finite dimensional spaces $X$ and $Y$ with the norms $\|.\|_X$ and $\|.\|_Y$, respectively. In the context of this paper, the abstract spaces $X,Y$ represent either subspaces of $\mathbb R^n$ or the space $\mathbb R^{3\times 3}_{sym}$.
\begin{defi}
Let $F:\ X\rightarrow Y$ be locally Lipschitz function in a neighborhood of some $x\in X$ and $\partial F(x)$ denote the generalized Jacobian in the Clarke sense \cite{Cl83}. We say that $F$ is strongly semismooth at $x$ if 
\begin{enumerate}
\item[$(i)$] $F$ is directionally differentiable at $x$,
\item[$(ii)$] for any $h\in X$, $h\rightarrow0$, and any $V\in\partial F(x+h)$,
\begin{equation}
F(x+h)-F(x)-Vh=O(\|h\|^2_X).
\label{semi_1}
\end{equation} 
\end{enumerate}
\label{def_semi}
\end{defi}
Notice that the estimate (\ref{semi_1}) is called the quadratic approximate property in \cite{Su01} or the strong $G$-semismoothness in \cite{Go04,PSS03}. In literature, there exists several equivalent definitions of strongly semismooth functions, see  \cite{QS93, Su01}. For example the condition (\ref{semi_1}) can be replaced with
\begin{equation}
F'(x+h;h)-F'(x;h)=O(\|h\|^2_X) \quad\forall h\in X,\;h\rightarrow0,\;x+h\in \mathcal D_F,
\label{semi_3}
\end{equation}
where $\mathcal D_F$ is the subset of $X$,  $F$ is Fr\'echet differentiable and $F'(x;h)$ denotes the directional derivative of $F$ at a point $x$ and a direction $h$.
We say that $F:\ X\rightarrow Y$ is strongly semismooth on an open set $\mathcal O\subset X$ if $F$ is strongly semismooth at every point of $\mathcal O$.

Since it is difficult to straightforwardly prove (\ref{semi_1}) or (\ref{semi_3}), we summarize several auxilliary results. Firstly, piecewise smooth ($PC^1$) functions with Lipschitz continuous derivatives of selected functions belong among strongly semismooth functions \cite{FP03}. Especially, we mention the max-function, $\xi: \ \mathbb R\rightarrow\mathbb R$, $\xi(s)=\max\{0,s\}=s^+$. Further, scalar product, sum, compositions of strongly semismooth functions are strongly semismooth. Finally, we will use the following version of the implicit function theorem \cite{Cl83, Su01, MSZ05}.
\begin{theorem}
Let $\mathcal  I:\ Y\times X\rightarrow X$ be a locally Lipschitz function in a neighborhood of $(\bar y, \bar x)$, which solve $\mathcal  I(\bar y, \bar x)=0$. Let $\partial_x \mathcal  I(y,x)$ 
denote the generalized derivatives of $\mathcal  I$ at $(y,x)$ with respect to the variables $x$.
If $\partial_x \mathcal  I(\bar y,\bar x)$ is of maximal rank, i.e. the following implication holds,
\begin{equation}
\mathcal  I^o_x\triangle x=0, \;\mathcal  I^o_x\in\partial_x \mathcal  I(\bar y,\bar x)\quad\Longrightarrow\quad \triangle x=0,
\label{full_rank}
\end{equation} 
then there exists an open neighborhood $\mathcal  O_{\bar y}$ of $\bar y$ and a function $F: O_{\bar y}\rightarrow X$ such that $F$ is locally Lipschitz continuous in $\mathcal  O_{\bar y}$, $F(\bar y)=\bar x$ and for every $y$ in $\mathcal  O_{\bar y}$, $\mathcal  I(y,F(y))=0$.

Moreover, if $\mathcal  I$ is strongly semismooth at $(\bar y, \bar x)$, then $F$ strongly semismooth at $\bar y$.
\label{implicit}
\end{theorem}

The semismoothness of constitutive operators in elastoplasticity has been studied e.g. in \cite{B97, GrVa09, SaWi11, Sy09, Sy14, CKSV14}. Namely in \cite{SaWi11, Sy14}, one can find an abstract framework how to investigate it for operators in an implicit form. However, the framework cannot be straightforwardly used for models investigated in this paper. Therefore, we introduce the following auxilliary result.

\begin{proposition}
Let $p(\cdot)$, $\mbf s(\cdot)$, $\varrho(\cdot)$, $\theta(\cdot)$, $\mbf n(\cdot)$, $r_e(\cdot)$ and $\varrho_e(\cdot)$ be the functions introduced in Section \ref{sec_invariants}. Further, ler $\hat p, \hat\varrho:\mathbb R^{3\times 3}_{sym}\rightarrow\mathbb R$ be strongly semismooth functions and assume that $\hat \varrho$ vanishes for any $\mbf\sigma\in\mathbb R^{3\times 3}_{sym}$, $\varrho(\mbf\sigma)= 0$. Define,
$$
\hat\varrho_e(\mbf\sigma):=\left\{
\begin{array}{ll}
\hat\varrho(\mbf\sigma)r_e(\cos\theta(\mbf\sigma)), & \varrho(\mbf\sigma)\neq0,\\
0, & \varrho(\mbf\sigma)=0,
\end{array}\right.,\quad
\mbf S(\mbf\sigma):=\left\{
\begin{array}{ll}
\hat p(\mbf\sigma)\mbf I+\hat\varrho(\mbf\sigma)\mbf n(\mbf\sigma), & \varrho(\mbf\sigma)\neq0,\\
\hat p(\mbf\sigma)\mbf I, & \varrho(\mbf\sigma)=0.
\end{array}\right.
$$
Then the functions $\varrho$, $\varrho_e$, $\hat\varrho_e$ and $\mbf S$ are strongly semismooth in $\mathbb R^{3\times 3}_{sym}$.
\label{prop_semi}
\end{proposition}

\begin{proof}
Since the functions $\mbf n(\cdot)$ and $r_e(\cos\theta(\cdot))$ are bounded and have Lipschitz continuous derivatives in $\{\mbf\sigma\in\mathbb R^{3\times 3}_{sym}\ |\; \varrho(\mbf\sigma)\neq0\}$,  it is easy to see that the functions $\varrho$, $\varrho_e$, $\hat\varrho_e$ and $\mbf S$ are locally Lipschitz continuous in $\mathbb R^{3\times 3}_{sym}$ and strongly semismooth for any $\mbf\sigma\in\mathbb R^{3\times 3}_{sym}$, $\varrho(\mbf\sigma)\neq0$. 
Therefore, it remains to show strong semismoothness at $\mbf\sigma\in\mathbb R^{3\times 3}_{sym}$, $\varrho(\mbf\sigma)= 0$. To this end, we show that (\ref{semi_3}) holds for $\varrho$, $\varrho_e$, $\hat\varrho_e$ and $\mbf S$ at such $\mbf\sigma$. Let $\mbf\tau\in\mathbb R^{3\times 3}_{sym}$ be such that $\varrho(\mbf\tau)\neq0$. Then 
$$\varrho(\mbf\sigma+\epsilon\mbf\tau)=\epsilon\varrho(\mbf\tau), \;\;
\mbf n(\mbf\sigma+\epsilon\mbf\tau)=\mbf n(\mbf\tau), \;\; \theta(\mbf\sigma+\epsilon\mbf\tau)=\theta(\mbf\tau) \quad \forall \epsilon>0.$$
Hence $\varrho'(\mbf\sigma+\mbf\tau;\mbf\tau)-\varrho'(\mbf\sigma;\mbf\tau)=0$, $\varrho'_e(\mbf\sigma+\mbf\tau;\mbf\tau)-\varrho'_e(\mbf\sigma;\mbf\tau)=0$ and 
\begin{eqnarray*}
\hat\varrho'_e(\mbf\sigma+\mbf\tau;\mbf\tau)-\hat\varrho'_e(\mbf\sigma;\mbf\tau)&=&[\hat\varrho'(\mbf\sigma+\mbf\tau;\mbf\tau)-\hat\varrho'(\mbf\sigma;\mbf\tau)]\mbf r_e(\cos\theta(\mbf\tau)){=}O(\|\mbf\tau\|^2),\\
\mbf S'(\mbf\sigma+\mbf\tau;\mbf\tau)-\mbf S'(\mbf\sigma;\mbf\tau)&=&[\hat p'(\mbf\sigma+\mbf\tau;\mbf\tau)-\hat p'(\mbf\sigma;\mbf\tau)]\mbf I+[\hat\varrho'(\mbf\sigma+\mbf\tau;\mbf\tau)-\hat\varrho'(\mbf\sigma;\mbf\tau)]\mbf n(\mbf\tau){=}O(\|\mbf\tau\|^2),
\end{eqnarray*}
since the functions $\hat p$, $\hat\varrho$ satisfy (\ref{semi_3}) by the assumption.
\end{proof}
The function $\mbf S$ introduced in Proposition \ref{prop_semi} has the same scheme as a mapping between trial and unknown stress tensors for models introduced in Section \ref{sec_DP}--\ref{sec_general}. Here, the trial stress is represented by $\mbf\sigma$ and the unknown stress is in the form $\hat{\mbf\sigma}=\hat p(\mbf\sigma)\mbf I+\hat\varrho(\mbf\sigma)\mbf n(\mbf\sigma)$. Therefore, it is sufficient to prove only semismoothness of the scalar functions $\hat p$, $\hat\varrho$ representing invariants of the unknown stress tensor. The semismoothness of $\hat \varrho_e$ has been derived to prove Theorem \ref{th_semi_JG}.


\section{The Drucker-Prager model}
\label{sec_DP}

\subsection{Constitutive initial value problem}
\label{sec_model}

We consider the  elastoplastic problem containing the Drucker-Prager criterion,
a nonassociative plastic flow rule and a nonlinear isotropic hardening. Within a
thermodynamical framework with internal variables, we introduce the
corresponding constitutive initial value problem, see \cite{NPO08}:

\begin{enumerate}
\item {\it Additive decomposition of the infinitesimal strain tensor
$\mbf\varepsilon$ on elastic  and plastic  parts}:
\begin{equation}
\mbf\varepsilon=\mbf{\varepsilon}^e+\mbf{\varepsilon}^p.
\label{split}
\end{equation}
\item {\it Linear isotropic elastic law between the stress and the elastic strain}:
\begin{equation}
\mbf\sigma=\mathbb D_e:\mbf{\varepsilon}^e={K}(\mbf I:\mbf{\varepsilon}^e)\mbf I+2G\mathbb I_{dev}:\mbf{\varepsilon}^e,\quad \mathbb D_e=K\mbf I\otimes\mbf I+2G\mathbb I_{dev},
\label{elastic_law}
\end{equation}
where $K,G>0$ denotes the bulk, and shear moduli, respectively.
 \item {\it Non-linear isotropic hardening}:
\begin{equation}
\kappa=H(\bar\varepsilon^p).
\label{hardening_law}
\end{equation}
Here $\bar\varepsilon^p\in\mathbb R_+$ denotes an isotropic (scalar) hardening
variable, $\kappa\in\mathbb R_+$ is the corresponding thermodynamical force and
$H: \mathbb R_+\rightarrow\mathbb R_+$ is a nondecreasing, strongly semismooth function satisfying $H(0)=0$. 
\item {\it  Drucker-Prager yield function}: 
\begin{equation}
f(\mbf\sigma,\kappa)=\hat f(p(\mbf\sigma),\varrho(\mbf\sigma),\kappa)=\sqrt{\frac{1}{2}}\varrho+\eta p-\xi (c_0+\kappa).
\label{yield_function}
\end{equation}
Here, the parameters $\eta,\xi>0$ are usually calculated from the friction angle using a sufficient approximation of the Mohr-Coulomb yield surface and $c_0>0$ denotes the initial cohesion. 
\item {\it Plastic pseudo-potential}.
\begin{equation}
g(\mbf\sigma)=\hat g(p(\mbf\sigma),\varrho(\mbf\sigma))=\sqrt{\frac{1}{2}}\varrho+\bar\eta p.
\label{potential_function}
\end{equation}
Here $\bar \eta>0$ denotes a parameter depending on the dilatancy angle.
\item {\it Nonassociative plastic flow rule} 
\begin{equation}
\dot{\mbf{\varepsilon}}^p\in\dot\lambda\partial g(\mbf\sigma),
\label{flow_rule}
\end{equation}
where $\dot\lambda\geq0$ is a multiplier and $\partial g(\mbf\sigma)$ denotes the subdifferential of the convex function $g$ at $\mbf\sigma$. Using (\ref{inv_der1}), (\ref{rho_subgrad}) and (\ref{potential_function}), the flow rule (\ref{flow_rule}) can be written as
\begin{equation}
\dot{\mbf{\varepsilon}}^p=\dot\lambda\left(\sqrt{\frac{1}{2}}\hat{\mbf n}+\frac{\bar\eta}{3}\mbf I\right),\quad \hat{\mbf n}\in\partial\varrho(\mbf\sigma).
\label{flow_rule2}
\end{equation}
Consequently by (\ref{split}), (\ref{elastic_law}) and (\ref{rho_subgrad_prop}),
\begin{equation}
\dot{\mbf{\sigma}}=\mathbb D_e:(\dot{\mbf{\varepsilon}}-\dot{\mbf{\varepsilon}}^p)=\mathbb D_e:\dot{\mbf{\varepsilon}}-\dot\lambda\left(G\sqrt{2}\hat{\mbf n}+K\bar\eta\mbf I\right).
\label{flow_rule3}
\end{equation}
\item {\it Associative hardening law:}
\begin{equation}
\dot{\bar{\varepsilon}}^p=-\dot\lambda\frac{\partial f(\mbf\sigma,\kappa)}{\partial\kappa}=\dot\lambda\xi.
\end{equation}
\item {\it Loading/unloding criterion:}
\begin{equation}
\dot\lambda\geq0,\quad f(\mbf\sigma,\kappa)\leq0,\quad \dot\lambda f(\mbf\sigma,\kappa)=0.
\label{loading_unloading}
\end{equation}
\end{enumerate}

Then the elastoplastic constitutive initial value problem reads as follows:
\textit{Given the history of the strain tensor $\mbf\varepsilon=\mbf\varepsilon(t)$, $t\in[t_0, t_{\max}]$, and the initial values
$\mbf{\varepsilon}^p(t_0)=\mbf{\varepsilon}^p, \; \bar\varepsilon^p(t_0)=\bar\varepsilon^p_0.$
Find $(\mbf\sigma(t),\mbf{\varepsilon}^p(t), \bar\varepsilon^p(t))$ such that 
\begin{equation}
\left.
\begin{array}{l}
\mbf\sigma=\mathbb D_e:(\mbf{\varepsilon}-\mbf{\varepsilon}^p),\\
\dot{\mbf{\sigma}}=\mathbb D_e:\dot{\mbf{\varepsilon}}-\dot\lambda\left(G\sqrt{2}\hat{\mbf n}+K\bar\eta\mbf I\right),\quad \hat{\mbf n}\in\partial\varrho(\mbf\sigma),\\
\dot{\bar{\varepsilon}}^p=\dot\lambda \xi,\\
\dot\lambda\geq0,\;\; f(\mbf{\sigma},H(\bar\varepsilon^p))\leq0,\;\; \dot\lambda f(\mbf{\sigma},H(\bar\varepsilon^p))=0.
\end{array}
\right\}
\label{CIVP_DP}
\end{equation}
hold for each instant $t\in[t_0,t_{\max}]$}.

\subsection{Implicit Euler discretization of CIVP}

We discretize CIVP using the implicit Euler method. To this end we assume a partition $$0=t_0<t_1<\ldots<t_k<\ldots<t_N=t_{\max}$$ of the pseudo-time interval and fix a step $k$. For the sake of brevity, we omit the index $k$ and write $\mbf{\sigma}:=\mbf{\sigma}(t_k)$, $\mbf\varepsilon:=\mbf\varepsilon(t_k)$, $\mbf\varepsilon^p:=\mbf\varepsilon^p(t_k)$ and $\bar\varepsilon^p=\bar\varepsilon^p(t_k)$. Further, we define the following trial variables: $\bar{\varepsilon}^{p,tr}:=\bar{\varepsilon}^p(t_{k-1})$, $\mbf\varepsilon^{e,tr}=\mbf\varepsilon(t_{k})-\mbf{\varepsilon}^p(t_{k-1})$ and $\mbf{\sigma}^{tr}:=\mathbb D_e:\mbf\varepsilon^{e,tr}$.
Then the discrete elastoplastic constitutive problem for the $k$-step reads as follows: \textit{Given $\mbf{\varepsilon}$, $\mbf{\sigma}^{tr}$ and $\bar{\varepsilon}^{p,tr}$. Find $\mbf{\sigma}$,  $\bar{\varepsilon}^p$ and $\triangle\lambda$ satisfying:}
\begin{equation}
\left.
\begin{array}{c}
\mbf{\sigma}=\mbf{\sigma}^{tr}-\triangle\lambda\left(G\sqrt{2}\hat{\mbf n}+K\bar\eta\mbf I\right),\quad \hat{\mbf n}\in\partial\varrho(\mbf{\sigma}),\\
\bar{\varepsilon}^p=\bar{\varepsilon}^{p,tr}+\triangle\lambda\xi,\\
\triangle\lambda\geq0,\quad f(\mbf{\sigma},H(\bar\varepsilon^p))\leq0,\quad \triangle\lambda f(\mbf{\sigma},H(\bar\varepsilon^p))=0.
\end{array}
\right\}
\label{k-step_problem}
\end{equation}
Notice that the remaining input parameter for the next step, $\mbf{\varepsilon}^p(t_k)$, can be computed using the formula $\mbf{\varepsilon}^p=\mbf{\varepsilon}-\mathbb
D_e^{-1}:\mbf{\sigma}$ after finding a solution to problem (\ref{k-step_problem}).

\subsection{Solution of the incremental problem}
\label{subsec_const_sol}

We standardly use the elastic predictor/plastic corrector method for solving (\ref{k-step_problem}).

\medskip\noindent
{\bf Elastic predictor} applies when
\begin{equation}
f(\mbf{\sigma}^{tr},H(\bar{\varepsilon}^{p,tr}))\leq0.
\label{trial_admissibility}
\end{equation}
Then the triplet
\begin{equation}
\mbf{\sigma}=\mbf{\sigma}^{tr},\quad \bar{\varepsilon}^p=\bar{\varepsilon}^{p,tr}, \quad \triangle\lambda=0
\end{equation}
 is the solution to (\ref{k-step_problem}).

\medskip\noindent
{\bf Plastic corrector} applies when (\ref{trial_admissibility}) does not hold. Then $\triangle\lambda>0$ and (\ref{k-step_problem}) reduces into
\begin{equation}
\left.
\begin{array}{c}
\mbf{\sigma}=\mbf{\sigma}^{tr}-\triangle\lambda\left(G\sqrt{2}\hat{\mbf n}+K\bar\eta\mbf I\right),\quad \hat{\mbf n}\in\partial\varrho(\mbf{\sigma}),\\
\bar{\varepsilon}^p=\bar{\varepsilon}^{p,tr}+\triangle\lambda\xi,\\
f(\mbf{\sigma},H(\bar\varepsilon^p))=0.
\end{array}
\right\}
\label{k-step_problem_plast}
\end{equation}
Since the functions $f$ and $g$ depend on $\mbf\sigma$ only through the variables $\varrho$ and $p$, it is natural to reduce a number of uknowns in problem (\ref{k-step_problem_plast}). To this end, we split  (\ref{k-step_problem_plast})$_1$ into the deviatoric and volumetric parts:
\begin{eqnarray}
\mbf{s}&=&\mbf{s}^{tr}-\triangle\lambda G\sqrt{2}\hat{\mbf n}, \quad \hat{\mbf n}\in\partial\varrho(\mbf{\sigma}), \label{flow_dev}\\
{p}&=&{p^{tr}}-\triangle\lambda K\bar\eta,
\end{eqnarray}
where $\mbf{s}^{tr}$, ${p^{tr}}$ denotes the deviatoric stress, and the hydrostatic stress related to $\mbf{\sigma}^{tr}$, respectively. Using (\ref{inv_der1})$_3$, the equality (\ref{flow_dev}) yields
\begin{equation}
\mbf{s}^{tr}=\left\{
\begin{array}{r c l}
\left(1+\frac{\triangle\lambda G\sqrt{2}}{\varrho}\right)\mbf{s}&\mbox{if}& \varrho> 0,\\[6pt]
\triangle\lambda G\sqrt{2}\hat{\mbf n}&\mbox{if}& \varrho= 0.
\end{array}\right.
\label{flow_s_k}
\end{equation}
Denote $\varrho^{tr}=\|\mbf{s}^{tr}\|$ and recall that $\|\hat{\mbf n}\|\leq1$ for $\varrho=0$ by (\ref{rho_subgrad}). Then from (\ref{flow_s_k}) we obtain
\begin{equation}
\left\{\begin{array}{r c l}
\varrho=\varrho^{tr}-\triangle\lambda G\sqrt{2} &\mbox{if}& \varrho> 0,\\
0\geq\varrho^{tr}-\triangle\lambda G\sqrt{2} &\mbox{if}& \varrho= 0.
\end{array}\right.
\label{flow_rho}
\end{equation}
Following the arguments developed in Section~1.1, we now rewrite
(\ref{flow_rho}) as follows:
\begin{equation}
\varrho=\left(\varrho^{tr}-\triangle\lambda G\sqrt{2}\right)^+=\max\left\{0;\; \varrho^{tr}-\triangle\lambda G\sqrt{2}\right\}.
\label{flow_rho_improve}
\end{equation}
Notice that (\ref{flow_rho}) and (\ref{flow_rho_improve}) are equivalent.
Further from (\ref{flow_s_k})$_1$, we standardly have:
\begin{equation}
\mbf{n}=\frac{\mbf{s}}{\varrho}=\frac{\mbf{s}^{tr}}{\varrho^{tr}}=\mbf{n}^{tr}\quad\mbox{if}\quad \varrho>0.
\end{equation}
The following theorem summarizes and completes the derived results.

\begin{theorem}
Let $f(\mbf{\sigma}^{tr},H(\bar{\varepsilon}^{p,tr}))>0$. If $(\mbf{\sigma},\bar{\varepsilon}^p,\triangle\lambda)$ is a solution to problem (\ref{k-step_problem_plast}) and $p=\mbf I: \mbf{\sigma}/3$, $\mbf{s}=\mathbb I_{dev}:\mbf{\sigma}$, $\varrho=\|\mbf{s}\|$, then $(p,\varrho,\bar{\varepsilon}^p,\triangle\lambda)$ is a solution to the following system:
\begin{equation}
\left.
\begin{array}{c}
{p}=p^{tr}-\triangle\lambda K\bar\eta,\\
\varrho=\left(\varrho^{tr}-\triangle\lambda G\sqrt{2}\right)^+,\\
\bar{\varepsilon}^p=\bar{\varepsilon}^{p,tr}+\triangle\lambda\xi,\\
\hat f(p,\varrho,H(\bar\varepsilon^p))=0.
\end{array}
\right\}
\label{k-step_problem_plast_red}
\end{equation}
Conversely, if $(p,\varrho,\bar{\varepsilon}^p,\triangle\lambda)$ is a solution to (\ref{k-step_problem_plast_red}) then $(\mbf{\sigma},\bar{\varepsilon}^p,\triangle\lambda)$ is the solution to (\ref{k-step_problem_plast}) where
\begin{equation}
\mbf{\sigma}=\left\{
\begin{array}{c c l}
\mbf{\sigma}^{tr}-\triangle\lambda\left(G\sqrt{2}\mbf{n}^{tr}+K\bar\eta\mbf I\right) &\mbox{if}& \varrho^{tr}>\triangle\lambda G\sqrt{2},\\
\left({p^{tr}}-\triangle\lambda K\bar\eta\right)\mbf I &\mbox{if}& \varrho^{tr}\leq\triangle\lambda G\sqrt{2}.\\
\end{array}
\right.
\end{equation}
\label{lem_auxilliary_problem}
\end{theorem}

Notice that the knowledge of the subdifferential of $\varrho$ enables us to formulate the plastic corrector problem as a unique system of nonlinear equations in comparison to the current technique introduced in \cite{NPO08}. Moreover, one can eliminate the unknowns $p,\varrho,\bar{\varepsilon}^p$ similarly as for the current return-mapping scheme of this model. Inserting of (\ref{k-step_problem_plast_red})$_{1-3}$ into (\ref{k-step_problem_plast_red})$_4$ leads to the nonlinear equation $q_{tr}(\triangle\lambda)=0$ where
\begin{equation}
q_{tr}(\gamma):=q(\gamma;p^{tr},\varrho^{tr},\bar{\varepsilon}^{p,tr})=\sqrt{\frac{1}{2}}\left(\varrho^{tr}-\gamma G\sqrt{2}\right)^++\eta(p^{tr}-\gamma K\bar\eta)-\xi\left(c_0+H(\bar{\varepsilon}^{p,tr}+\gamma\xi)\right),\quad\gamma\in\mathbb R_+,
\label{q_DP}
\end{equation}
using (\ref{yield_function}). We have the following solvability result.

\begin{theorem}
Let $f(\mbf{\sigma}^{tr},H(\bar{\varepsilon}^{p,tr}))>0$. Then there exists a unique solution, $\triangle\lambda>0$, of the equation $q_{tr}(\triangle\lambda)=0$. Furthermore, problems (\ref{k-step_problem_plast_red}), (\ref{k-step_problem_plast}) and (\ref{k-step_problem}) have also unique solutions. 

In addition, if $q_{tr}\left(\varrho^{tr}/G\sqrt{2}\right)<0$ then $\triangle\lambda\in(0,\varrho^{tr}/G\sqrt{2})$ and $\varrho>0$. Conversely, if $q_{tr}\left(\varrho^{tr}/G\sqrt{2}\right)$ $\geq0$ then $\triangle\lambda\geq\varrho^{tr}/G\sqrt{2}$ and $\varrho=0$.
\label{th_solvability_DP}
\end{theorem}

\begin{proof}
From (\ref{q_DP}) and the assumptions on $H$, it is readily seen that $q_{tr}$ is a continuous and decreasing function. Further, $q_{tr}(\gamma)\rightarrow-\infty$ as $\gamma\rightarrow+\infty$ and $q_{tr}(0)=f(\mbf{\sigma}^{tr},H(\bar{\varepsilon}^{p,tr}))>0$. Therefore, the equation $q_{tr}(\triangle\lambda)=0$ has just one solution in $\mathbb R_+$. If $q_{tr}\left(\varrho^{tr}/G\sqrt{2}\right)<0$ then $\triangle\lambda\in(0,\varrho^{tr}/G\sqrt{2})$. Otherwise, $\triangle\lambda\geq\varrho^{tr}/G\sqrt{2}$. The rest of the proof follows from Theorem \ref{lem_auxilliary_problem} and the elastic prediction.
\end{proof}

The second part of Theorem \ref{th_solvability_DP} is very useful from the computational point of view: one can a priori decide whether return to the smooth portion of the yield surface happens or not. This is the main difference in comparison with the current return-mapping scheme. The improved return-mapping scheme reads as follows.

\medskip\noindent
{\bf  Return to the smooth portion} 
\begin{enumerate}
\item {\it Necessary and sufficient condition}:  $q_{tr}(0)>0$ and $q_{tr}\left(\varrho^{tr}/G\sqrt{2}\right)<0$.
\item {\it Find} $\triangle\lambda\in(0,\varrho^{tr}/G\sqrt{2})$:
\begin{equation}
\sqrt{\frac{1}{2}}\left(\varrho^{tr}-\triangle\lambda G\sqrt{2}\right)+\eta\left(p^{tr}-\triangle\lambda K\bar\eta\right)-\xi\left(c_0+H(\bar{\varepsilon}^{p,tr}+\triangle\lambda\xi)\right)=0.
\label{lambda_problem1}
\end{equation}
\item {\it Set}
\begin{equation}
\mbf{\sigma}=\mbf{\sigma}^{tr}-\triangle\lambda\left(G\sqrt{2}\mbf{n}^{tr}+K\bar\eta\mbf I\right),\quad \bar{\varepsilon}^p=\bar{\varepsilon}^{p,tr}+ \triangle\lambda\xi.
\label{sigma_smooth}
\end{equation}
\end{enumerate}

\medskip\noindent
{\bf  Return to the apex} 
\begin{enumerate}
\item {\it Necessary and sufficient condition}: $q_{tr}\left(\varrho^{tr}/G\sqrt{2}\right)\geq0$. 
\item {\it Find} $\triangle\lambda\geq\varrho^{tr}/G\sqrt{2}$:
\begin{equation}
\eta\left(p^{tr}-\triangle\lambda K\bar\eta\right)-\xi\left(c_0+H(\bar{\varepsilon}^{p,tr}+\triangle\lambda\xi)\right)=0.
\label{lambda_problem2}
\end{equation}
\item {\it Set}
\begin{equation}
\mbf{\sigma}=\left({p^{tr}}-\triangle\lambda K\bar\eta\right)\mbf I,\quad \bar{\varepsilon}^p=\bar{\varepsilon}^{p,tr}+ \triangle\lambda\xi. 
\label{sigma_apex}
\end{equation}
\end{enumerate}
Nonlinear equations (\ref{lambda_problem1}) and (\ref{lambda_problem2}) can be
solved by the Newton method. Then it is natural to use the initial choice
$\triangle\lambda^0=0$, $\triangle\lambda^0=\varrho^{tr}/G\sqrt{2}$ for
(\ref{lambda_problem1}), and (\ref{lambda_problem2}), respectively. In case of perfect plasticity, $H=0$, or linear hardening, $H(\bar{\varepsilon}^p)=\tilde H\bar{\varepsilon}^p$, $\tilde H=const.$, equations (\ref{lambda_problem1}) and (\ref{lambda_problem2}) are linear and thus $\triangle\lambda$ can be found in the closed form.


\subsection{Stress-strain and consistent tangent operators}
\label{sec_stress-strain}

Solving the problem  (\ref{k-step_problem}), we obtain a nonlinear and implicit operator between the stress tensor, $\mbf\sigma=\mbf\sigma(t_k)$, and the strain tensor, $\mbf\varepsilon=\mbf\varepsilon(t_k)$. The stress-strain operator, $\mbf T$, also depends on $\mbf{\varepsilon}^p(t_{k-1})$ and $\bar{\varepsilon}^p(t_{k-1})$ through the trial variables. To emphasize this fact we write $\mbf{\sigma}:=\mbf T(\mbf{\varepsilon};\mbf{\varepsilon}^p(t_{k-1}),\bar\varepsilon^p(t_{k-1}))$. From the results introduced in Section \ref{subsec_const_sol}, we have $\mbf T(\mbf{\varepsilon};\mbf{\varepsilon}^p(t_{k-1}),\bar\varepsilon^p(t_{k-1}))=\mbf S(\mbf{\sigma}^{tr},\bar\varepsilon^{p,tr})$, where
\begin{equation}
\mbf S(\mbf{\sigma}^{tr},\bar\varepsilon^{p,tr})=\left\{
\begin{array}{c c c}
\mbf{\sigma}^{tr} &\mbox{if}& q_{tr}(0)\leq0,\\
\mbf{\sigma}^{tr}-\triangle\lambda\left(G\sqrt{2}\mbf{n}^{tr}+K\bar\eta\mbf I\right) &\mbox{if}&q_{tr}(0)>0,\;q_{tr}\left({\varrho^{tr}}/{G\sqrt{2}}\right)<0,\\
\left({p^{tr}}-\triangle\lambda K\bar\eta\right)\mbf I &\mbox{if}& q_{tr}\left({\varrho^{tr}}/{G\sqrt{2}}\right)\geq0,
\end{array}
\right.
\label{S_def}
\end{equation}
where $\triangle\lambda$ is the solution to (\ref{lambda_problem1}), (\ref{lambda_problem2}) in (\ref{S_def})$_2$, and (\ref{S_def})$_3$, respectively, i.e. $q_{tr}(\triangle\lambda)=0$.

\begin{theorem}
The function $\mbf T$ is strongly semismooth in $\mathbb R^{3\times3}_{sym}$ with respect to $\mbf\varepsilon$.
\label{th_semi_DP}
\end{theorem}
\begin{proof}
We use the framework introduced in Section \ref{sub_semi}. Consider the function $\triangle\lambda:=\triangle\lambda(\mbf{\sigma}^{tr},\bar\varepsilon^{p,tr})$ satisfying $\triangle\lambda=0$ if $q_{tr}(0)\leq0$, otherwise $q_{tr}(\triangle\lambda)=0$. Applying Theorem \ref{implicit} on the implicit function $q_{tr}$, one can easily find that the function $\triangle\lambda$ is strongly semismooth. Consequently, the functions $\hat p(\mbf{\sigma}^{tr},\bar\varepsilon^{p,tr})={p^{tr}}-(\triangle\lambda)^+ K\bar\eta$ and $\hat \varrho(\mbf{\sigma}^{tr},\bar\varepsilon^{p,tr})=\left(\varrho^{tr}-(\triangle\lambda)^+ G\sqrt{2}\right)^+$ are strongly semismooth. Since $\mbf S(\mbf{\sigma}^{tr},\bar\varepsilon^{p,tr})=\hat p(\mbf{\sigma}^{tr},\bar\varepsilon^{p,tr})\mbf I+\hat\varrho(\mbf{\sigma}^{tr},\bar\varepsilon^{p,tr})\mbf n(\mbf{\sigma}^{tr})$, we obtain strong semismoothness of the functions $\mbf S$ and $\mbf T$ using Proposition \ref{prop_semi}.
\end{proof}

Notice that $\mbf T$ is not smooth if $q_{tr}(0)=0$ or $q_{tr}\left({\varrho^{tr}}/{G\sqrt{2}}\right)=0$ or if $H$ has not derivative at $\bar{\varepsilon}^{p,tr}+\triangle\lambda\xi$. We introduce the derivative $\partial\mbf\sigma/\partial\mbf\varepsilon$ under the assumption that any of these conditions does not hold. Set $H_1:=H'(\bar{\varepsilon}^{p,tr}+\triangle\lambda\xi)$.
Using (\ref{inv_der1}), (\ref{inv_der2}), (\ref{elastic_law}) and the chain rule, we obtain the following auxilliary derivatives:
$$\frac{\partial \mbf\sigma^{tr}}{\partial\mbf{\varepsilon}}=\mathbb D_e,\quad \frac{\partial p^{tr}}{\partial\mbf{\varepsilon}}=K\mbf I,\quad\frac{\partial\mbf{s}^{tr}}{\partial\mbf{\varepsilon}}=2G\mathbb I_{dev},\quad \frac{\partial \varrho^{tr}}{\partial\mbf{\varepsilon}}=2G\mbf{n}^{tr},\quad\frac{\partial\mbf{n}^{tr}}{\partial\mbf{\varepsilon}}=\frac{2G}{\varrho^{tr}}\left(\mathbb I_{dev}-\mbf{n}^{tr}\otimes\mbf{n}^{tr}\right).$$
We distinguish three possible cases:
\begin{enumerate}
\item Let $q_{tr}(0)<0$ (\textit{elastic response}). Then clearly,
\begin{equation}
\frac{\partial\mbf{\sigma}}{\partial\mbf{\varepsilon}}=\mathbb D_e.
\label{deriv_elast}
\end{equation}
\item Let $q_{tr}(0)>0$ and $q_{tr}\left({\varrho^{tr}}/{G\sqrt{2}}\right)<0$
(\textit{return to the smooth surface}). Then the derivative of
(\ref{sigma_smooth}) reads
\begin{equation*}
\frac{\partial\mbf{\sigma}}{\partial\mbf{\varepsilon}}=\mathbb D_e-\triangle\lambda\frac{2G^2\sqrt{2}}{\varrho^{tr}}\left(\mathbb I_{dev}-\mbf{n}^{tr}\otimes\mbf{n}^{tr}\right)-(G\sqrt{2}\mbf{n}^{tr}+K\bar\eta\mbf I)\otimes\frac{\partial\triangle\lambda}{\partial\mbf{\varepsilon}}.
\end{equation*}
Applying the implicit function theorem on (\ref{lambda_problem1}), we obtain
\begin{equation*}
\frac{\partial\triangle\lambda}{\partial\mbf{\varepsilon}}=\frac{G\sqrt{2}\mbf{n}^{tr}+\eta K\mbf I}{G+K\eta\bar\eta+\xi^2H_1}.
\end{equation*}
Hence,
\begin{equation}
\frac{\partial\mbf{\sigma}}{\partial\mbf{\varepsilon}}=\mathbb D_e-\triangle\lambda\frac{2G^2\sqrt{2}}{\varrho^{tr}}\left(\mathbb I_{dev}-\mbf{n}^{tr}\otimes\mbf{n}^{tr}\right)-(G\sqrt{2}\mbf{n}^{tr}+K\bar\eta\mbf I)\otimes\frac{G\sqrt{2}\mbf{n}^{tr}+\eta K\mbf I}{G+K\eta\bar\eta+\xi^2H_1}.
\label{deriv_smooth}
\end{equation}
\item Let $q_{tr}\left({\varrho^{tr}}/{G\sqrt{2}}\right)>0$ (\textit{return to the apex}). Then the derivative of (\ref{sigma_apex}) yields
\begin{equation*}
\frac{\partial\mbf{\sigma}}{\partial\mbf{\varepsilon}}=K\left(\mbf I\otimes \mbf I-\bar\eta\mbf I\otimes\frac{\partial\triangle\lambda}{\partial\mbf{\varepsilon}}\right).
\end{equation*}
Applying the implicit function theorem on (\ref{lambda_problem2}), we obtain
\begin{equation*}
\frac{\partial\triangle\lambda}{\partial\mbf{\varepsilon}}=\frac{\eta K}{K\eta\bar\eta+\xi^2H_1}\mbf I.
\end{equation*}
Hence,
\begin{equation}
\frac{\partial\mbf{\sigma}}{\partial\mbf{\varepsilon}}=K\left(1-\frac{K\eta\bar\eta}{K\eta\bar\eta+\xi^2H_1}\right)\mbf I\otimes \mbf I.
\label{deriv_apex}
\end{equation}
\end{enumerate}
The derivatives (\ref{deriv_elast})--(\ref{deriv_apex}) define the consistent tangent operator, $\mathbb T^o$. It is readily seen that the tangent operator is symmetric if $\bar\eta=\eta$, i.e. for the associative plasticity. For purposes of Section \ref{sec_realization}, it is useful to extend the definition of $\mathbb T^o$ for nondifferential points. For example, one can write
\begin{equation}
\mathbb T^o(\mbf{\varepsilon};\mbf{\varepsilon}^p(t_{k-1}),\bar\varepsilon^p(t_{k-1}))=\left\{
\begin{array}{c c c}
\mathbb D_e &\mbox{if}& q_{tr}(0)\leq0,\\[2pt]
(\ref{deriv_smooth}) &\mbox{if}&q_{tr}(0)>0,\;q_{tr}\left({\varrho^{tr}}/{G\sqrt{2}}\right)<0,\\[2pt]
(\ref{deriv_apex}) &\mbox{if}& q_{tr}\left({\varrho^{tr}}/{G\sqrt{2}}\right)\geq0,
\end{array}
\right.
\end{equation}
where $H_1$ in (\ref{deriv_smooth}), (\ref{deriv_apex}) is the derivative from left of $H$ at
$\bar{\varepsilon}^{p,tr}+\triangle\lambda\xi$. Notice that $$\mathbb T^o(\mbf{\varepsilon};\mbf{\varepsilon}^p(t_{k-1}),\bar\varepsilon^p(t_{k-1}))\in\partial_{\mbf\varepsilon} \mbf T(\mbf{\varepsilon};\mbf{\varepsilon}^p(t_{k-1}),\bar\varepsilon^p(t_{k-1})).$$


\section{A simplified version of the Jir\'asek-Grassl model}
\label{sec_JG}

The Jir\'asek-Grassl model was introduced in \cite{GJ06}.  It is a plastic-damage model proposed for complex modelling of concrete failure. The model has been further developed. For example,  Unteregger and Hofstetter \cite{UH15} have improved a hardening law and used the model in rock mechanics. For the sake of simplicity, we only consider a perfect plastic part of this model to illustrate the suggested idea and improve the implicit return-mapping scheme. The whole plastic part of the Jir\'asek-Grassl model can be included to an abstract model studied in the next section.

\subsection{Constitutive problem and its solution}

The perfect plastic model contains the yield function proposed in \cite{MW95}:
\begin{equation}
f(\mbf\sigma)=\hat f(p(\mbf\sigma),\varrho(\mbf\sigma),\varrho_e(\mbf\sigma))=\frac{3}{2}\left(\frac{\varrho}{\bar f_c}\right)^2+m_0\left(\frac{\varrho_e}{\sqrt{6}\bar f_c}+\frac{p}{\bar f_c}\right)-1,
\label{yield_function_JG}
\end{equation}
where $m_0$ is the friction parameter and $\bar f_c$ is the uniaxial compressive strength. The invariants $p$, $\varrho$ and $\varrho_e=\varrho r_e(\cos\theta)$ were introduced in Section \ref{sec_preliminaries}. Notice that the couple $(p^a,\varrho^a)=(\bar f_c/m_0,0)$ defines the apex of the yields surface generated by the function $\hat f$. The yield surface is not smooth only at this apex. Scheme of the yield surface can be found in \cite{GJ06}.

Further, the following plastic pseudo-potential is considered \cite{GJ06}:
\begin{equation}
g(\mbf\sigma)=\hat g(p(\mbf\sigma),\varrho(\mbf\sigma))=\frac{3}{2}\left(\frac{\varrho}{\bar f_c}\right)^2+\frac{m_0\varrho}{\sqrt{6}\bar f_c}+\frac{m_g(p)}{\bar f_c},
\label{potential_JG}
\end{equation}
where 
\begin{equation}
m_g(p)=A_gB_g\bar f_c e^{\frac{p-\bar f_t/3}{B_g\bar f_c}},\qquad A_g, B_g, \bar f_c, \bar f_t>0.
\label{m_g}
\end{equation}
The subdifferential of $g$ consists of the following directions:
$$\frac{1}{3}\hat g_V(p(\mbf{\sigma}),\varrho(\mbf{\sigma}))\mbf I+\hat g_\varrho(p(\mbf{\sigma}),\varrho(\mbf{\sigma}))\hat{\mbf n},\quad \hat{\mbf n}\in\partial\varrho(\mbf\sigma),$$
where $\partial\varrho(\mbf\sigma)$ is defined by (\ref{rho_subgrad}) and
 $$\hat g_V(p,\varrho):=\frac{\partial \hat g}{\partial p}=\frac{m_g'(p)}{\bar f_c},\quad \hat g_\varrho(p,\varrho):=\frac{\partial \hat g}{\partial \varrho}=\frac{3\varrho}{\bar f_c^2}+\frac{m_0}{\sqrt{6}\bar f_c}, \quad m_g'(p)=A_ge^{\frac{p-\bar f_t/3}{B_g\bar f_c}}.$$

The $k$-step of the incremental constitutive problem received by the implicit Euler method reads as follows. \textit{Given $\mbf{\varepsilon}:=\mbf{\varepsilon}(t_k)$ and $\mbf{\sigma}^{tr}:=\mathbb D_e:(\mbf{\varepsilon}(t_k)-\mbf{\varepsilon}^p(t_{k-1}))$. Find $\mbf{\sigma}=\mbf{\sigma}(t_k)$ and $\triangle\lambda$ satisfying:
\begin{equation}
\left.
\begin{array}{c}
\mbf{\sigma}=\mbf{\sigma}^{tr}-\triangle\lambda\left[K\frac{m_g'(p)}{\bar f_c}\mbf I+2G\left(\frac{3\varrho}{\bar f_c^2}+\frac{m_0}{\sqrt{6}\bar f_c}\right)\hat{\mbf n}\right],\quad \hat{\mbf n}\in\partial\varrho(\mbf\sigma),\\[5pt]
\triangle\lambda\geq0,\quad \hat f(p(\mbf{\sigma}),\varrho(\mbf{\sigma}),\varrho_e(\mbf{\sigma}))\leq0,\quad \triangle\lambda\hat f(p(\mbf{\sigma}),\varrho(\mbf{\sigma}),\varrho_e(\mbf{\sigma}))=0.
\end{array}
\right\}
\label{k-step_problem_JG}
\end{equation}}

We solve this problem again by the elastic predictor/plastic corrector method. Within the plastic correction, we define the trial variables $\mbf s^{tr}$, $p^{tr}$, $\varrho^{tr}$, $\mbf n^{tr}=\mbf s^{tr}/\varrho^{tr}$, $\theta^{tr}$, $r_e^{tr}=r_e(\cos\theta^{tr})$ and $\varrho_e^{tr}$ associated with $\mbf{\sigma}^{tr}$ and obtain the following result. 
\begin{theorem}
Let $\hat f(p^{tr},\varrho^{tr},\varrho_e^{tr})>0$. If $(\mbf{\sigma},\triangle\lambda)$ is a solution to problem (\ref{k-step_problem_JG}) and $p=\mbf I: \mbf{\sigma}/3$, $\mbf{s}=\mathbb I_{dev}:\mbf{\sigma}$, $\varrho=\|\mbf{s}\|$, then $(p,\varrho,\triangle\lambda)$ is a solution to the following system:
\begin{equation}
\left.
\begin{array}{c}
{p}=p^{tr}-\triangle\lambda K\frac{m_g'(p)}{\bar f_c},\\[5pt]
\varrho=\left[\varrho^{tr}-\triangle\lambda 2G\left(\frac{3\varrho}{\bar f_c^2}+\frac{m_0}{\sqrt{6}\bar f_c}\right)\right]^+,\\[5pt]
\hat f(p,\varrho,\varrho r_e^{tr})=0.
\end{array}
\right\}
\label{k-step_problem_plast_JG}
\end{equation}
Conversely, if $(p,\varrho,\triangle\lambda)$ is a solution to (\ref{k-step_problem_plast_JG}) then $(\mbf{\sigma},\triangle\lambda)$ is the solution to (\ref{k-step_problem_JG}) where
\begin{equation}
\mbf{\sigma}=\left\{
\begin{array}{c c l}
p \mbf I+\varrho\mbf {n}^{tr} &\mbox{if}& \varrho^{tr}>\triangle\lambda 2G\left(\frac{3\varrho}{\bar f_c^2}+\frac{m_0}{\sqrt{6}\bar f_c}\right),\\[5pt]
p\mbf I &\mbox{if}& \varrho^{tr}\leq\triangle\lambda 2G\left(\frac{3\varrho}{\bar f_c^2}+\frac{m_0}{\sqrt{6}\bar f_c}\right).\\
\end{array}
\right.
\end{equation}
\label{lem_auxilliary_problem_JG}
\end{theorem}
\begin{proof}
To prove Theorem \ref{lem_auxilliary_problem_JG} we use the same technique as in Section \ref{subsec_const_sol}. It is based on the splitting the stress tensor on the deviatoric and volumetric parts, and on using linear dependence between $\mbf s$ and $\mbf s^{tr}$ to reduce a number of unknowns.  In particular, we have
$$\mbf s^{tr}=\left(1+\triangle\lambda 2G\left(\frac{3\varrho}{\bar f_c^2}+\frac{m_0}{\sqrt{6}\bar f_c}\right)\frac{1}{\varrho}\right)\mbf{s}$$
for $\varrho>0$. Consequently, we obtain (\ref{k-step_problem_plast_JG})$_2$, $\mbf n=\mbf n^{tr}$ and also $\theta=\theta^{tr}$ for $\varrho>0$ using (\ref{theta}).
Finally, notice that $\varrho r_e^{tr}\rightarrow 0$ as $\varrho^{tr}\rightarrow0_+$. Indeed, $\varrho\rightarrow0$ as $\varrho^{tr}\rightarrow0_+$ and the function $r_e(\cos(.))$ is bounded.
\end{proof}

Analogously to the Drucker-Prager model, one can analyze existence and uniqueness of a solution to problem (\ref{k-step_problem_plast_JG}), and a priori decide whether the return to the smooth portion of the yield surface happens or not. To this end, we define implicit functions $\hat p_{tr}:\gamma\mapsto p_\gamma$ and $\hat \varrho_{tr}:\gamma\mapsto \varrho_\gamma$ such that
$$p_\gamma+\gamma K\frac{m_g'(p_\gamma)}{\bar f_c} -p^{tr}=0,\quad \varrho_\gamma-\left[\varrho^{tr}-\gamma 2G\left(\frac{3\varrho_\gamma}{\bar f_c^2}+\frac{m_0}{\sqrt{6}\bar f_c}\right)\right]^+=0,$$
respectively, for any $\gamma\geq0$. The following lemma is a consequence of the implicit function theorem.

\begin{lem}
The functions $\hat p_{tr}$ and $\hat\varrho_{tr}$ are well-defined in $\mathbb R_+$. Further, $\hat p_{tr}$ is smooth and decreasing in $\mathbb R_+$,  $\hat\varrho_{tr}$ is decreasing in the interval $\left[0, \frac{\sqrt{6}\bar f_c\varrho^{tr}}{2Gm_0}\right)$ and its closed form reads as follows:
\begin{equation}
\hat\varrho_{tr}(\gamma)=\frac{1}{1+\gamma\frac{6G}{\bar f_c^2}}\left(\varrho^{tr}-\gamma\frac{2Gm_0}{\sqrt{6}\bar f_c}\right)^+\quad\forall \gamma\geq0.
\label{hat_varrho}
\end{equation}
\label{lem_JG}
\end{lem}

Now, consider the function $q_{tr}(\gamma):=q(\gamma;p^{tr},\varrho^{tr})$,
\begin{equation}
q_{tr}(\gamma)=\hat f(\hat p_{tr}(\gamma),\hat\varrho_{tr}(\gamma),\hat\varrho_{tr}(\gamma)r_e^{tr})=\frac{3}{2}\left(\frac{\hat\varrho_{tr}(\gamma)}{\bar f_c}\right)^2+m_0\left(\frac{\hat\varrho_{tr}(\gamma)r_e^{tr}}{\sqrt{6}\bar f_c}+\frac{\hat p_{tr}(\gamma)}{\bar f_c}\right)-1.
\label{q_JG}
\end{equation}

\begin{theorem}
Let $\hat f(p^{tr},\varrho^{tr},\varrho_e^{tr})>0$. Then there exists a unique solution, $\triangle\lambda>0$, of the equation $q_{tr}(\triangle\lambda)=0$. Furthermore, problems (\ref{k-step_problem_plast_JG}) and (\ref{k-step_problem_JG}) have also unique solutions. 

In addition, if $q_{tr}\left(\sqrt{6}\bar f_c\varrho^{tr}/2Gm_0\right)<0$ then $\triangle\lambda\in(0,\sqrt{6}\bar f_c\varrho^{tr}/2Gm_0)$ and $\varrho>0$. Conversely, if $q_{tr}\left(\sqrt{6}\bar f_c\varrho^{tr}/2Gm_0\right)\geq0$ then $\triangle\lambda\geq\sqrt{6}\bar f_c\varrho^{tr}/2Gm_0$ and $\varrho=0$.
\label{th_solvability_JG}
\end{theorem}

\begin{proof}
Since $\hat \varrho_{tr}>0$ and $\hat \varrho'_{tr}<0$ in $\left[0, \frac{\sqrt{6}\bar f_c\varrho^{tr}}{2Gm_0}\right)$, the functions $\hat\varrho_{tr}$, $\hat\varrho_{tr}^2$ are decreasing in this interval. For $\gamma\geq\frac{\sqrt{6}\bar f_c\varrho^{tr}}{2Gm_0}$, these functions vanish. Therefore,
from (\ref{q_JG}) and Lemma \ref{lem_JG}, it is follows that $q_{tr}$ is a continuous and decreasing function in $\mathbb R_+$.  Furthermore, $q_{tr}(\gamma)\rightarrow-\infty$ as $\gamma\rightarrow+\infty$ and $q(0)=\hat f(p^{tr},\varrho^{tr},\varrho_e^{tr})>0$. Hence, the equation $q_{tr}(\triangle\lambda)=0$ has a unique solution in $\mathbb R_+$. If $q_{tr}\left(\sqrt{6}\bar f_c\varrho^{tr}/2Gm_0\right)<0$ then $\triangle\lambda\in(0,\sqrt{6}\bar f_c\varrho^{tr}/2Gm_0)$. Otherwise, $\triangle\lambda\geq\sqrt{6}\bar f_c\varrho^{tr}/2Gm_0$. The rest of the proof follows from Theorem \ref{lem_auxilliary_problem_JG} and the elastic prediction.
\end{proof}

Although the function $q_{tr}$ is implicit the decision criterion introduced in Theorem \ref{th_solvability_JG} can be found in closed form.
\begin{lem}
\begin{equation}
q_{tr}\left(\frac{\sqrt{6}\bar f_c\varrho^{tr}}{2Gm_0}\right)\geq0 \quad\mbox{if and only if}\quad p^{tr}-\frac{\sqrt{6}K}{2G}\frac{m_g'(p^a)}{m_0}\varrho^{tr}-p^a\geq0,\quad p^a=\frac{\bar f_c}{m_0}.
\label{criterion_JG}
\end{equation}
\end{lem}
\begin{proof}
Since $\hat\varrho_{tr}\left(\frac{\sqrt{6}\bar f_c\varrho^{tr}}{2Gm_0}\right)=0$,  
$$q_{tr}\left(\frac{\sqrt{6}\bar f_c\varrho^{tr}}{2Gm_0}\right)\stackrel{(\ref{q_JG})}{=}\frac{m_0}{\bar f_c}\hat p_{tr}\left(\frac{\sqrt{6}\bar f_c\varrho^{tr}}{2Gm_0}\right)-1\geq 0.$$
Hence, $p^{crit}:=\hat p_{tr}\left(\frac{\sqrt{6}\bar f_c\varrho^{tr}}{2Gm_0}\right)\geq p^a$. Using the definitions of $\hat p_{tr}$ and $m_g$, we have
$$0=p^{crit}-\frac{\sqrt{6}\varrho^{tr}}{2Gm_0}Km_g'(p^{crit})-p^{tr}\geq p^{a}-\frac{\sqrt{6}\varrho^{tr}}{2Gm_0}Km_g'(p^{a})-p^{tr}.$$
\end{proof}

By Theorem \ref{lem_auxilliary_problem_JG} and Theorem \ref{th_solvability_JG}, the return-mapping scheme reads as follows.

\medskip\noindent
{\bf  Return to the smooth portion} 
\begin{enumerate}
\item {\it Necessary and sufficient condition}:  $q_{tr}(0)>0$ and $q_{tr}\left(\sqrt{6}\bar f_c\varrho^{tr}/2Gm_0\right)<0$.
\item {\it Find} $p\in\mathbb R$, $\varrho>0$ and $\triangle\lambda\in(0,\sqrt{6}\bar f_c\varrho^{tr}/2Gm_0)$:
\begin{equation}
\left.
\begin{array}{c}
{p}+\triangle\lambda K\frac{m_g'(p)}{\bar f_c} -p^{tr}=0,\\[5pt]
\varrho+\triangle\lambda 2G\left(\frac{3\varrho}{\bar f_c^2}+\frac{m_0}{\sqrt{6}\bar f_c}\right)-\varrho^{tr}=0,\\[5pt]
\frac{3}{2}\left(\frac{\varrho}{\bar f_c}\right)^2+m_0\left(\frac{\varrho r_e^{tr}}{\sqrt{6}\bar f_c}+\frac{p}{\bar f_c}\right)-1=0.
\end{array}
\right\}
\label{lambda_problem1_JG}
\end{equation}
\item {\it Set}
\begin{equation}
\mbf{\sigma}=p \mbf I+\varrho\mbf {n}^{tr}.
\label{sigma_smooth_JG}
\end{equation}
\end{enumerate}

\medskip\noindent
{\bf  Return to the apex} 
\begin{enumerate}
\item {\it Necessary and sufficient condition}: $q_{tr}\left(\sqrt{6}\bar f_c\varrho^{tr}/2Gm_0\right)\geq0$. 
\item {\it Set}
\begin{equation}
p=\frac{\bar f_c}{m_0},\quad \varrho=0,\quad \mbf{\sigma}=p\mbf I,\quad \triangle\lambda=\frac{\bar f_c}{Km_g'(p)}(p^{tr}-p). 
\label{sigma_apex_JG}
\end{equation}
\end{enumerate}
The system (\ref{lambda_problem1_JG}) of nonlinear equations can be solved by the Newton method with the initial choice $p^0=p^{tr}$, $\varrho^0=\varrho^{tr}$, $\triangle\lambda^0=0$. It was shown that the system has a unique solution subject to $q_{tr}(0)>0$ and $q_{tr}\left(\sqrt{6}\bar f_c\varrho^{tr}/2Gm_0\right)<0$. Without these conditions, one cannot guarantee existence and uniqueness of the solution to (\ref{lambda_problem1_JG}).

\subsection{Stress-strain and consistent tangent operators}
\label{subsec_operator_JG}

Solving the problem  (\ref{k-step_problem_JG}), we obtain a nonlinear and implicit operator between the stress tensor, $\mbf\sigma=\mbf\sigma(t_k)$, and the strain tensor, $\mbf\varepsilon=\mbf\varepsilon(t_k)$. The stress-strain operator, $\mbf T$, also depends on $\mbf{\varepsilon}^p(t_{k-1})$ through the trial stress. To emphasize this fact we write $\mbf{\sigma}:=\mbf T(\mbf{\varepsilon};\mbf{\varepsilon}^p(t_{k-1}))$. We have 
\begin{equation}
\mbf T(\mbf{\varepsilon};\mbf{\varepsilon}^p(t_{k-1}))=\left\{
\begin{array}{c c l}
\mbf{\sigma}^{tr} &\mbox{if}& q_{tr}(0)\leq0,\\
p\mbf I+\varrho \mbf n^{tr} &\mbox{if}&q_{tr}(0)>0,\;q_{tr}\left(\sqrt{6}\bar f_c\varrho^{tr}/2Gm_0\right)<0,\\
\frac{\bar f_c}{m_0}\mbf I &\mbox{if}& \qquad \qquad \quad q_{tr}\left(\sqrt{6}\bar f_c\varrho^{tr}/2Gm_0\right)\geq0,
\end{array}
\right.
\label{T_def}
\end{equation}
where the function $q_{tr}$ is defined by (\ref{q_JG}) and $p$, $\varrho$ are components of the solution to (\ref{lambda_problem1_JG}). 

\begin{theorem}
The function $\mbf T$ is strongly semismooth in $\mathbb R^{3\times3}_{sym}$ with respect to $\mbf\varepsilon$.
\label{th_semi_JG}
\end{theorem}

\begin{proof}
Consider the function $\triangle\lambda=\triangle\lambda(\mbf{\sigma}^{tr})$ satisfying $\triangle\lambda=0$ if $q_{tr}(0)\leq0$, otherwise $q_{tr}(\triangle\lambda)=0$. To apply Theorem \ref{implicit} on the implicit function $q_{tr}(\gamma):=q(\gamma;p^{tr},\varrho^{tr})$, it is necessary to show that $q$ is strongly semismooth w.r.t. the variables $\gamma,p^{tr},\varrho^{tr}$. This follows from (\ref{hat_varrho}) and Proposition \ref{prop_semi}. The rest of the proof coincides with the proof of Theorem \ref{th_semi_DP}.
\end{proof}

If $q_{tr}(0)=0$ or $q_{tr}\left(\sqrt{6}\bar f_c\varrho^{tr}/2Gm_0\right)=0$ then $\mbf T$ is not smooth. We derive the derivative $\partial\mbf\sigma/\partial\mbf\varepsilon$ under the assumption that any of these conditions does not hold. If $q_{tr}(0)<0$ (elastic response) then $\partial\mbf\sigma/\partial\mbf\varepsilon=\mathbb D_e$. If $q_{tr}\left(\sqrt{6}\bar f_c\varrho^{tr}/2Gm_0\right)>0$ (return to the apex) then $\partial\mbf\sigma/\partial\mbf\varepsilon=\mathbb O$ (vanishes).

Let $q_{tr}(0)>0$ and $q_{tr}\left(\sqrt{6}\bar f_c\varrho^{tr}/2Gm_0\right)<0$, i.e., return to the smooth portion happens. Then the derivative $\partial\mbf\sigma/\partial\mbf\varepsilon$ can be found as follows.
\begin{enumerate}
\item Find the solution $(p,\varrho,\triangle\lambda)$ to (\ref{lambda_problem1_JG}).
\item Use (\ref{inv_der1}), (\ref{inv_der2}), (\ref{elastic_law}) and the chain rule and compute:
$$\frac{\partial p^{tr}}{\partial\mbf{\varepsilon}}=K\mbf I,\quad
\frac{\mbf{s^{tr}}}{\partial\mbf{\varepsilon}}=2G\mathbb I_{dev},\quad
\frac{\partial \varrho^{tr}}{\partial\mbf{\varepsilon}}=2G\mbf{n^{tr}},\quad
\frac{\partial\mbf{n^{tr}}}{\partial\mbf{\varepsilon}}=\frac{2G}{\varrho^{tr}}\left(\mathbb I_{dev}-\mbf{n^{tr}}\otimes\mbf{n^{tr}}\right),$$ 
$$\frac{\partial \theta^{tr}}{\partial\mbf{\varepsilon}}=\frac{2G\sqrt{6}}{\varrho^{tr}\sin(3\theta^{tr})}\left[(\mbf{n^{tr}}\otimes(\mbf{n^{tr}})^3)\mbf I-\mathbb I_{dev}(\mbf{n^{tr}})^2\right],\quad \frac{\partial r_e^{tr}}{\partial\mbf{\varepsilon}}=-r_e'(\cos\theta^{tr})\sin\theta^{tr}\frac{\partial \theta^{tr}}{\partial\mbf{\varepsilon}}.$$
\item Compute: 
\begin{equation}
\left(
\begin{array}{c}
\frac{\partial p}{\partial\mbf{\varepsilon}}\\[5pt]
\frac{\partial \varrho}{\partial\mbf{\varepsilon}}\\[5pt]
\frac{\partial \triangle\lambda}{\partial\mbf{\varepsilon}}
\end{array}
\right)=\left(
\begin{array}{c c c}
1+\triangle\lambda K\frac{m_g''(p)}{\bar f_c} & 0 & K\frac{m_g'(p)}{\bar f_c}\\
0 & 1+\triangle\lambda \frac{6G}{\bar f_c^2} & 2G\left(\frac{3\varrho}{\bar f_c^2}+\frac{m_0}{\sqrt{6}\bar f_c}\right)\\
\frac{m_0}{\bar f_c} & \frac{3\varrho}{\bar f_c^2}+\frac{m_0r_e^{tr}}{\sqrt{6}\bar f_c} & 0
\end{array}
\right)^{-1}\left(\begin{array}{c}
\frac{\partial p^{tr}}{\partial\mbf{\varepsilon}}\\[5pt]
\frac{\partial \varrho^{tr}}{\partial\mbf{\varepsilon}}\\[5pt]
-\frac{m_0\varrho}{\sqrt{6}\bar f_c}\frac{\partial r_e^{tr}}{\partial\mbf{\varepsilon}}
\end{array}\right).
\label{linearized_system}
\end{equation}
Notice that the matrix in (\ref{linearized_system}) arises from  linearization of (\ref{lambda_problem1_JG}) around the solution $(p,\varrho,\triangle\lambda)$. The matrix is invertible since its determinant is negative.
\item Compute
\begin{equation}
\frac{\partial\mbf\sigma}{\partial\mbf\varepsilon}=\mbf I\otimes\frac{\partial p}{\partial\mbf{\varepsilon}}+\mbf n^{tr}\otimes\frac{\partial \varrho}{\partial\mbf{\varepsilon}}+\varrho\frac{\partial\mbf{n^{tr}}}{\partial\mbf{\varepsilon}}.
\label{deriv_smooth_JG}
\end{equation}
\end{enumerate} 

For numerical purposes, we use the following generalized consistent tangent operator:
\begin{equation}
\mathbb T^o(\mbf{\varepsilon};\mbf{\varepsilon}^p(t_{k-1}))=\left\{
\begin{array}{c c l}
\mathbb D_e &\mbox{if}& q_{tr}(0)\leq0,\\[2pt]
(\ref{deriv_smooth_JG}) &\mbox{if}&q_{tr}(0)>0,\;q_{tr}\left(\sqrt{6}\bar f_c\varrho^{tr}/2Gm_0\right)<0,\\[2pt]
\mathbb O &\mbox{if}&  \qquad \qquad \quad q_{tr}\left(\sqrt{6}\bar f_c\varrho^{tr}/2Gm_0\right)\geq0.
\end{array}
\right.
\end{equation}


\section{An abstract model}
\label{sec_general}

The aim of this section is an extension of Theorem \ref{lem_auxilliary_problem} and \ref{lem_auxilliary_problem_JG} on a specific class of elastoplastic models that are usually formulated in the Haigh-Westergaard coordinates. We consider an abstract model containing the isotropic hardening and the plastic flow pseudo-potential.
\textit{Given the history of the strain tensor $\mbf\varepsilon=\mbf\varepsilon(t)$, $t\in[t_0, t_{\max}]$, and the initial values
$$\mbf{\varepsilon}^p(t_0)=\mbf{\varepsilon}^p, \;\; \bar\varepsilon^p(t_0)=\bar\varepsilon^p_0.$$
Find the generalized stress $(\mbf\sigma(t),\kappa(t))$ and the generalized strain $(\mbf{\varepsilon}^p(t), \bar\varepsilon^p(t))$ such that 
$$\left.
\begin{array}{l}
\mbf\varepsilon=\mbf{\varepsilon}^e+\mbf{\varepsilon}^p,\\
\mbf\sigma=\mathbb D_e:\mbf{\varepsilon}^e,\;\; \kappa=H(\bar\varepsilon^p),\\
\dot{\mbf{\varepsilon}}^p\in\dot\lambda \partial_\sigma g(\mbf\sigma,\kappa),\\
\dot{\bar{\varepsilon}}^p=\dot\lambda \ell(\mbf\sigma,\kappa),\\
\dot\lambda\geq0,\;\; f(\mbf\sigma,\kappa)\leq0,\;\; \dot\lambda f(\mbf\sigma,\kappa)=0.
\end{array}
\right\}$$
hold for each instant $t\in[t_0,t_{\max}]$}. 

Further, we have the following assumptions on ingredients of the model:
\begin{enumerate}
\item $f(\mbf\sigma,\kappa)=\hat f(p(\mbf\sigma),\varrho(\mbf\sigma),\varrho_e(\mbf\sigma),\kappa),$
where $\hat f$ is increasing with respect to $\varrho$ and $\tilde\varrho$, convex and continuously differentiable at least in vicinity of the yield surface. 
\item $g(\mbf\sigma,\kappa)=\hat g(p(\mbf\sigma),\varrho(\mbf\sigma),\kappa),$
where $\hat g$ is an increasing function with respect to $\varrho$, convex and twice continuously differentiable at least in vicinity of the yield surface. 
\item $H$ is a nondecreasing, continuous and strongly semismooth function satisfying $H(0)=0$.
\item $\ell(\mbf\sigma,\kappa)= \hat \ell(p(\mbf\sigma),\varrho(\mbf\sigma),\tilde\varrho(\mbf\sigma),\kappa)$ is a positive function.
\item Invariants $p$, $\varrho$, $\varrho_e$ and $\tilde\varrho$ are the same as in Section \ref{sec_preliminaries}.
\end{enumerate}
Notice that the assumptions on $\hat f$ and $\hat g$ guarantee convexity of $f$ and $g$ using properties of $r_e$ introduced in \cite{WW74}.
Let
$\hat g_V(p,\varrho):=\frac{\partial \hat g}{\partial p},\; \hat g_\varrho(p,\varrho):=\frac{\partial \hat g}{\partial \varrho}.$
Then one can write the plastic flow rule as follows:
\begin{equation*}
\dot{\mbf{\varepsilon}}^p=\dot\lambda\left[\hat g_V(p,\varrho)\mbf I/3+\hat g_\varrho(p,\varrho)\hat{\mbf n}\right],\quad \hat{\mbf n}\in\partial\varrho(\mbf\sigma).
\label{flow_rule2_abstract}
\end{equation*}

The $k$-th step of the incremental constitutive problem received by the implicit Euler method reads as follows. \textit{Given $\mbf{\varepsilon}:=\mbf{\varepsilon}(t_k)$, $\mbf{\sigma}^{tr}:=\mathbb D_e:(\mbf{\varepsilon}(t_k)-\mbf{\varepsilon}^p(t_{k-1}))$ and $\bar{\varepsilon}^{p,tr}:=\bar{\varepsilon}^p(t_{k-1})$. Find $\mbf{\sigma}=\mbf{\sigma}(t_k)$,  $\bar{\varepsilon}^{p}=\bar{\varepsilon}^{p}(t_k)$ and $\triangle\lambda$ satisfying:
\begin{equation}
\left.
\begin{array}{c}
\mbf{\sigma}=\mbf{\sigma}^{tr}-\triangle\lambda\left[K\hat g_V(p(\mbf{\sigma}),\varrho(\mbf{\sigma}))\mbf I+2G\hat g_\varrho(p(\mbf{\sigma}),\varrho(\mbf{\sigma}))\hat{\mbf n}\right],\quad \hat{\mbf n}\in\partial\varrho(\mbf\sigma),\\[3pt]
\bar{\varepsilon}^{p}=\bar{\varepsilon}^{p,tr}+\triangle\lambda \ell(\mbf\sigma,\kappa),\\[3pt]
\triangle\lambda\geq0,\quad f(\mbf{\sigma},H(\bar{\varepsilon}^{p}))\leq0,\quad \triangle\lambda f(\mbf{\sigma},H(\bar{\varepsilon}^{p}))=0.
\end{array}
\right\}
\label{k-step_problem_abstract}
\end{equation}}

If we use the elastic predictor/plastic corrector method then we derive the following straightforward extension of Theorem \ref{lem_auxilliary_problem} and Theorem \ref{lem_auxilliary_problem_JG} within the plastic correction.

\begin{theorem}
Let $f(\mbf{\sigma}^{tr},H(\bar{\varepsilon}^{p,tr}))>0$. If $(\mbf{\sigma},\bar{\varepsilon}^{p},\triangle\lambda)$ is a solution to problem (\ref{k-step_problem_abstract}) then $(p,\varrho,\bar{\varepsilon}^{p},\triangle\lambda)$, $\triangle\lambda>0$, is a solution to the following system:
\begin{equation}
\left.
\begin{array}{c}
{p}=p^{tr}-\triangle\lambda K\hat g_V(p,\varrho),\\[3pt]
\varrho=\left[\varrho^{tr}-\triangle\lambda2G\hat g_\varrho(p,\varrho)\right]^+,\\[3pt]
\bar{\varepsilon}^{p}=\bar{\varepsilon}^{p,tr}+\triangle\lambda \hat\ell\left(p,\varrho,\varrho\tilde r(\cos \theta^{tr})\right),\\[3pt]
\hat f\left(p,\varrho,\varrho r_e(\cos \theta^{tr}),H(\bar{\varepsilon}^{p})\right)=0.
\end{array}
\right\}
\label{k-step_problem_abstract_red}
\end{equation}

Conversely, if $(p,\varrho,\kappa,\triangle\lambda)$ is a solution to (\ref{k-step_problem_abstract_red}) then $(\mbf{\sigma},\kappa,\triangle\lambda)$ solves (\ref{k-step_problem_abstract}), where
\begin{equation}
\mbf{\sigma}=\left\{
\begin{array}{c c l}
p\mbf I +\varrho\mbf{n}^{tr}&\mbox{if}& \varrho>0,\\[3pt]
p\mbf I &\mbox{if}& \varrho=0.
\end{array}
\right.
\label{sigma_k_abstract}
\end{equation}
\label{lem_auxilliary_problem_abstract}
\end{theorem}

Notice that it is generally impossible to a priori decide about the type of the return as in the models introduced above. To be in accordance with the current approach introduced e.g. in \cite{GJ06} one can split (\ref{k-step_problem_abstract_red}) into the following two systems:
\begin{equation}
\left.
\begin{array}{lcl}
{p}+\triangle\lambda K\hat g_V(p,0)&=&p^{tr}\\
\bar{\varepsilon}^{p}-\triangle\lambda \hat \ell\left(p,0,0\right)&=&\bar{\varepsilon}^{p,tr}\\
\hat f\left(p,0,0,H(\bar{\varepsilon}^{p})\right)&=&0
\end{array}
\right\}
\quad\mbox{for } \varrho=0 \;\;\mbox{(\textit{return to the apex (apices)})},
\label{f_equality_apex}
\end{equation}
\begin{equation}
\left.
\begin{array}{lcl}
{p}+\triangle\lambda K\hat g_V(p,\varrho)&=&p^{tr}\\
\varrho+\triangle\lambda2G\hat g_\varrho(p,\varrho)&=&\varrho^{tr}\\
\bar{\varepsilon}^{p}-\triangle\lambda \hat \ell\left(p,\varrho,\varrho\tilde r(\cos \theta^{tr})\right)&=&\bar{\varepsilon}^{p,tr}\\
\hat f\left(p,\varrho,\varrho r_e(\cos \theta^{tr}),H(\bar{\varepsilon}^{p})\right)&=&0
\end{array}
\right\}
\quad\mbox{for } \varrho>0\;\;\mbox{(\textit{return to the smooth portion})},
\label{f_equality_smooth}
\end{equation}
and guess which of these systems provides an admissible solution. Beside the blind guessing, the current approach has another drawback: it can happen that (\ref{k-step_problem_abstract_red}) has a unique solution and mutually one of the systems (\ref{f_equality_apex}), (\ref{f_equality_smooth}) does not have any solution or have more than one solution. Therefore, we recommend to solve (\ref{k-step_problem_abstract_red}) directly by a nonsmooth version of the Newton method with the standard initial choice $p^0=p^{tr}$, $\varrho^0=\varrho^{tr}$, $\kappa^0=\kappa^{tr}$ and $\triangle\lambda^0=0$.


\section{Numerical realization of the incremental boundary value elastoplastic problem}
\label{sec_realization}

Consider an elasto-plastic body occupying  a bounded domain $\Omega\subseteq \mathbb{R}^3$ with the Lipschitz continuous boundary $\Gamma$. It is assumed that $\Gamma = \overline{\Gamma}_D \cup \overline{\Gamma}_N$, where $\Gamma_D$ and $\Gamma_N$ are open and disjoint sets. On $\Gamma_D$, the homogeneous Dirichlet boundary condition is prescribed. Surface tractions of density $\mbf f_t$ are applied on $\Gamma_N$ and the body is subject to a volume force $\mbf f_V$.

Notice that the above defined stress, strain and hardening variables depend on the spatial variable $\mbf x\in \Omega$, i.e. $\mbf{\sigma}_k=\mbf{\sigma}_k(\mbf x)$, etc. 
Let $\mathcal{V}:= \left\{\mbf v\in [H^1(\Omega)]^3\ |\;\mbf v = \mbf 0\ \mbox{on }\Gamma_D \right\}$
denote the space of kinematically admissible displacements. Under the infinitesimal small strain assumption, we have
$\mbf\varepsilon(\mbf v)=\frac{1}{2}\left(\nabla\mbf v+(\nabla\mbf v)^T\right),\; \mbf v\in \mathcal V.$

Substitution of the stress-strain operator $\mbf T$  into the principle of the virtual work leads to the following problem at the $k$-th step:
\begin{equation*}
(P_k)\quad \mbox{find } \mbf{u}_k\in\mathcal V: \quad \int_\Omega \mbf  T\left(\mbf\varepsilon(\mbf{u}_k);\mbf{\varepsilon}^p_{k-1},\bar\varepsilon_{k-1}^p\right):\mbf\varepsilon(\mbf v)\, \mbox{d}\Omega= \int_\Omega \mbf f_{V,k}.\mbf v\,d\Omega+\int_{\Gamma_N} \mbf f_{t,k}.\mbf v\,\mbox{d}\Gamma \quad\forall \mbf v\in \mathcal V,
\label{eqn}
\end{equation*} 
where $\mbf f_{V,k}$ and $\mbf f_{t,k}$ are the prescribed volume, and surface
forces at $t_k$, respectively. After finding a solution $\mbf u_k$, the
remaining unknown fields $\mbf{\varepsilon}^p_{k},\bar\varepsilon_{k}^p$
important for the next step can be computed at the level of integration points.
Problem $(P_k)$ can be standardly written as the operator equation in the dual space $\mathcal V'$ to $\mathcal V$: $\mathcal F_k(\mbf u_k)=\ell_k$, where
\begin{eqnarray*}
\langle\mathcal F_k(\mbf u),\mbf v\rangle&=&\int_\Omega \mbf T\left(\mbf\varepsilon(\mbf{u});\mbf{\varepsilon}^p_{k-1},\bar\varepsilon_{k-1}^p\right):\mbf\varepsilon(\mbf v)\, \mbox{d}\Omega\quad\forall \mbf u, \mbf v\in\mathcal V,\\
\langle \ell_k,\mbf v\rangle&=&\int_\Omega \mbf f_{V,k}.\mbf v\,d\Omega+\int_{\Gamma_N} \mbf f_{t,k}.\mbf v\,\mbox{d}\Gamma \quad\forall \mbf v\in \mathcal V.
\end{eqnarray*}
Since we plan to use the semismooth Newton method, we also introduce the operator $\mathcal K_k:\mathcal V\rightarrow\mathcal L(\mathcal V,\mathcal V')$ as follows:
$$\langle\mathcal K_k(\mbf u)\mbf v,\mbf w\rangle=\int_\Omega \mathbb T^o\left(\mbf\varepsilon(\mbf{u});\mbf{\varepsilon}^p_{k-1},\bar\varepsilon_{k-1}^p\right)\mbf\varepsilon(\mbf v):\mbf\varepsilon(\mbf w)\, \mbox{d}\Omega\quad\forall \mbf u, \mbf v, \mbf w\in\mathcal V.$$

To discretize the problem in space we use the finite element method. Then the space $\mathcal V$ is approximated by a finite dimensional one, $\mathcal V_h$. If linear simplicial elements are not used then it is also necessary to consider a suitable numerical quadrature on each element. Let $\mathcal F_{k,h}$, $\mathcal K_{k,h}$, $\ell_{k,h}$ denote the approximation of operators $\mathcal F_{k}$, $\mathcal K_{k}$, $\ell_{k}$, respectively, and $\mbf F_k:\mathbb R^n\rightarrow\mathbb R^n$, $\mbf K_k:\mathbb R^n\rightarrow\mathbb R^{n\times n}$, $\mbf l_k\in\mathbb R^n$ be their algebraic counterparts. Then the discretization of problem $(P_k)$ leads to the system of nonlinear equations, $\mbf{F}_{k}(\mbf{u}_{k})=\mbf{l}_{k}$, and 
the semismooth Newton method reads as follows:

\begin{algorithm}[Semismooth Newton method]
\hspace{0.2cm}
\begin{spacing}{1.2}
\begin{algorithmic}[1]
  \STATE initialization: $\mbf{u}_{k}^0=\mbf u_{k-1}$
  \FOR{$i=0,1,2,\ldots$}
    \STATE find $\delta \mbf{u}^{i}\in\mbf{V}$: $\;\mbf {K}_k(\mbf{u}_{k}^i)\delta  \mbf{u}^{i}=\mbf{l}_{k}-\mbf{F}_{k}(\mbf{u}_{k}^i)$
    \STATE compute $\mbf{u}_{k}^{i+1}=\mbf{u}_{k}^i+\delta \mbf{u}^{i}$
    \STATE {\bf if
    }{$\|\delta
    \mbf{u}^{i}\|/(\|\mbf{u}_{k}^{i+1}\|+\|\mbf{u}_{k}^i\|)\leq\epsilon_{Newton}$} {\bf then stop}
  \ENDFOR
  \STATE set $\mbf{u}_{k}=\mbf{u}_{k}^{i+1}$.
\end{algorithmic}
\end{spacing}
\end{algorithm}

If $\mbf T$ is strongly semismooth in $\mathbb R^{3\times 3}_{sym}$ then $\mbf F_k$ is strongly semismooth in $\mathbb R^n$. Notice that the strong semismoothness is an essential assumption for local quadratic convergence of this algorithm. In numerical examples introduced below, we observe local quadratic convergence when the tolerance is sufficiently small. In particular, we set $\epsilon_{Newton}=10^{-12}$.


\section{Numerical example - slope stability}
\label{sec_experiments}

The improved return-mapping schemes in combination with the semismooth Newton method have been partially implemented in codes SIFEL \cite{SIFEL} and MatSol \cite{MatSol-2009}. Here, for the sake of simplicity, we consider the slope stability benchmark \cite[Page 351]{NPO08} for the presented models, the Drucker-Prager (DP) and the Jirasek-Grassl (JG) ones. The benchmark is formulated as a plane strain problem. We focus on: a) incremental limit analysis and b) dependence of loading paths on element types and mesh density. For purposes of such an experiment, special MatLab codes have been prepared to be transparent. These experimental codes are available in \cite{Mcode} together with selected graphical outputs.

The geometry of the body is depicted in  Figure \ref{fig.mesh_P1} or \ref{fig.mesh_Q2}. The slope height is 10 m and its inclination is $45^\circ$. On the bottom, we assume that the body is fixed and, on the left and right sides,  zero normal displacements are prescribed. The body is subjected to self-weight. We set the specific weight $\rho g=20\, $kN/m$^3$ with $\rho$ being the mass density and $g$ the gravitational acceleration. Such a volume force is multiplied by a scalar factor, $\zeta$. The loading process starts from $\zeta=0$. The gravity load factor, $\zeta$, is then increased gradually until collapse occurs. The initial increment of the factor is set to 0.1. To illustrate loading responses we compute settlement at the corner point $A$ on the top of the slope depending on $\zeta$.

As in \cite[Page 351]{NPO08}, we set $E=20\,000\,$kPa, $\nu=0.49$, $\phi=20^{\circ}$ and $c=50\,$kPa, where $c$ denotes the cohesion for the perfect plastic model. Hence, $G=67\,114\, $kPa and $K=3\ 333\ 333\,$kPa. In comparison to \cite{NPO08}, we use the presented models instead of the Mohr-Coulomb one. The remaining parameters for these models will be introduced below.

We analyze the problem for linear triangular ($P1$) elements and eight-pointed quadrilateral ($Q2$) elements. In the latter case, $(3\times 3)$-point Gauss quadrature is used. For each element type, a hierarchy of four meshes with different densities is considered. The $P1$-meshes contain 3210, 12298, 48126, and 190121 nodal points, respectively. The $Q2$-meshes consist of 627, 2405, 9417, and 37265 nodal points, respectively. The coarsest meshes for $P1$ and $Q2$ elements are depicted in Figure \ref{fig.mesh_P1} and \ref{fig.mesh_Q2}. Let us complete that the mesh in Figure \ref{fig.mesh_P1} is uniform in vicinity of the slope and consists of right isoscales triangles with the same diagonal orientation. Further, it is worth mentioning that the $P1$-meshes are chosen much more finer in vicinity of the slope than their $Q2$-counterparts within the same level.

\begin{figure}[htbp]
\begin{minipage}[t]{0.47\textwidth}
  \center
   \includegraphics[width=\textwidth]{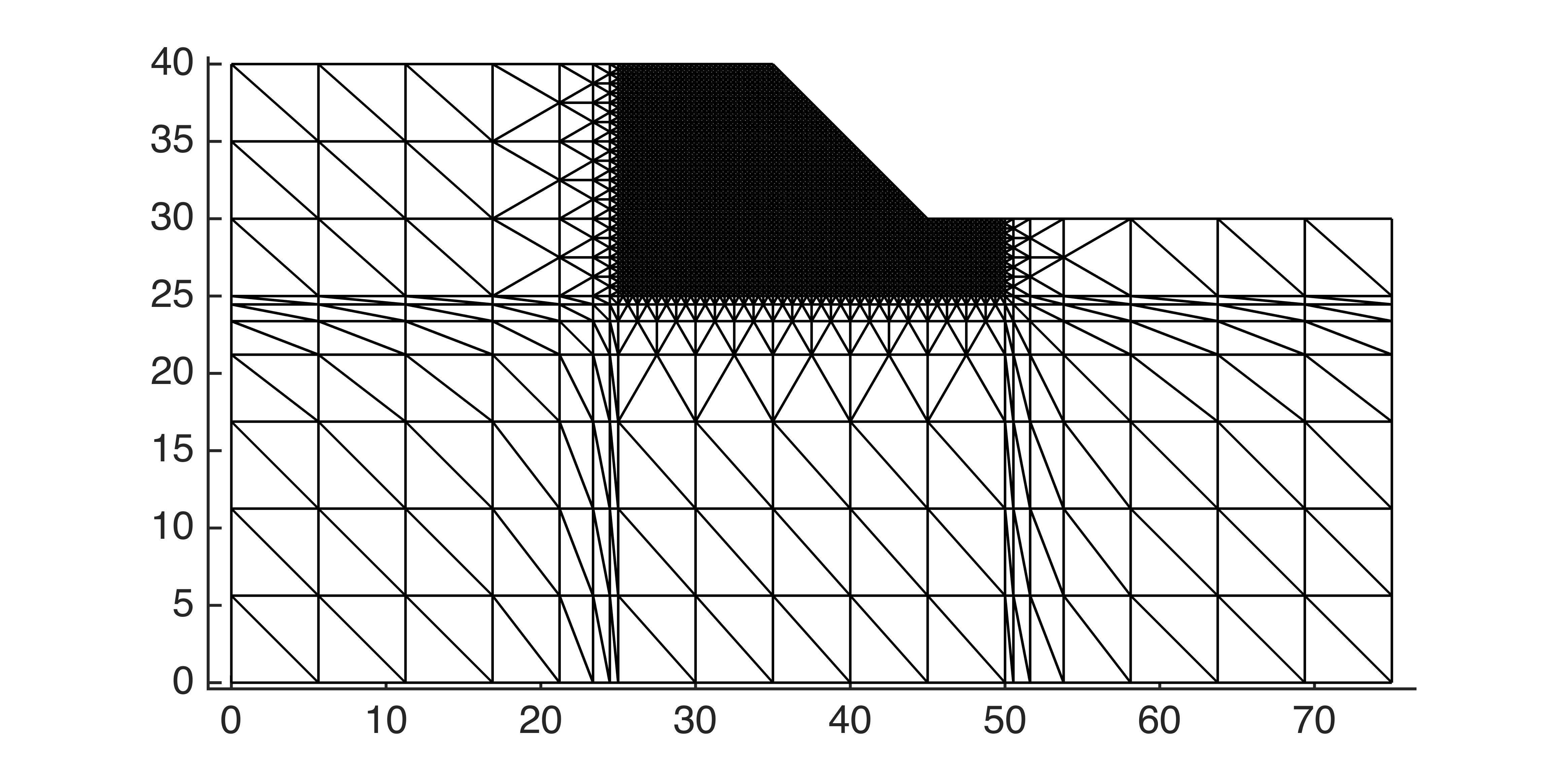}
   \caption{\small{The coarsest mesh for $P1$ elements.}}
   \label{fig.mesh_P1}
\end{minipage}
\hfill
\begin{minipage}[t]{0.47\textwidth}
  \center
  \includegraphics[width=\textwidth]{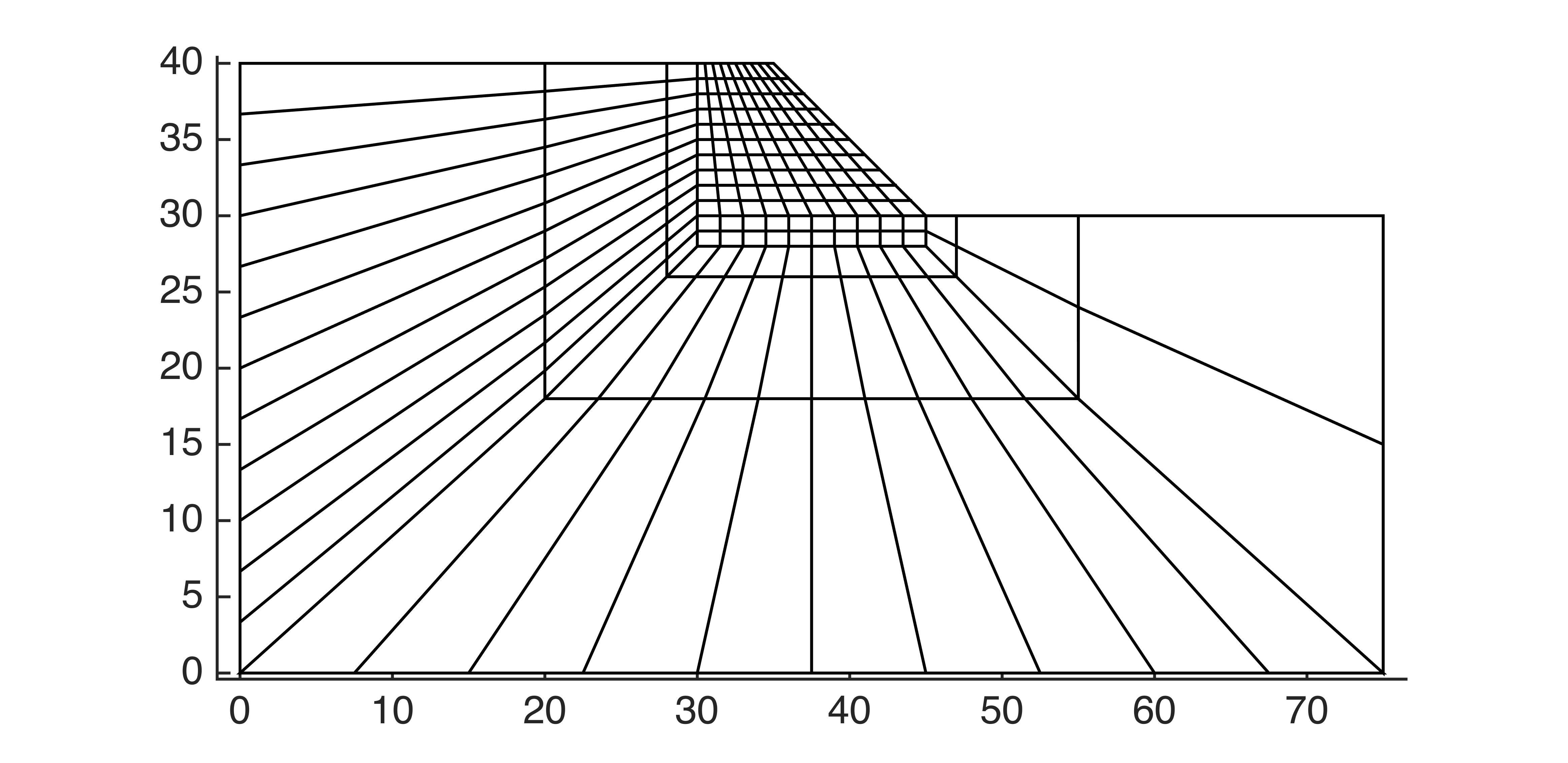}
   \caption{\small{The coarsest mesh for $Q2$ elements.}}
   \label{fig.mesh_Q2}
\end{minipage}
\end{figure}

\subsection{The Drucker-Prager model}

The Drucker-Prager parameters $\eta$, $\overline{\eta}$ and $\xi$ are computed from the friction angle, $\phi$, and the dilatancy angle, $\psi$, as follows \cite{NPO08}:
$$
  \eta = \frac{3\tan\phi}{\sqrt{9 + 12(\tan\phi)^2}},\ \ \overline{\eta} = \frac{3\tan\psi}{\sqrt{9 + 12(\tan\psi)^2}},\ \ \xi = \frac{3}{\sqrt{9 + 12(\tan\phi)^2}}.
$$
At first, we introduce results obtained for the model with associative perfect plasticity. In such a case, $\psi=\phi$, $c_0=c$ and $H=0\,$kPa. The received loading curves for the investigated meshes and elements are depicted in  Figure \ref{fig.load_path_h_P1} and \ref{fig.load_path_h_Q2}. Although $P1$-meshes are much finer, we observe more significant dependence of the curves on the mesh density for $P1$-elements than for $Q2$-elements. Also computed limit load factors are greater and tend more slowly to a certain value as $h\rightarrow0_+$ for $P1$-meshes than for $Q2$-meshes. The expected limit value  is 4.045 as follows from considerations introduced in \cite{CL90}. Using the finest $P1$ and $Q2$ meshes, we receive the values 4.230, and 4.056, respectively.

In general, higher order elements are recommended when a locking effect is expected. In this example, it can be caused due to the presence of the limit load and/or $\nu\approx 1/2$. On the other hand, the strong dependence on mesh density is influenced by other factors like mesh structure or choice of a model. For example, this dependence is not so significant for the Jirasek-Grassl model (see the next subsection).  Further, in \cite{HRS15}, there is theoretically justified and illustrated that the dependence of the limit load on the mesh density is minimal for bounded yield surfaces and that an approximation of unbounded yield surfaces by bounded ones (the truncation) leads to a lower bound of the limit load.

\begin{figure}[htbp]
\begin{minipage}[t]{0.47\textwidth}
  \center
   \includegraphics[width=\textwidth]{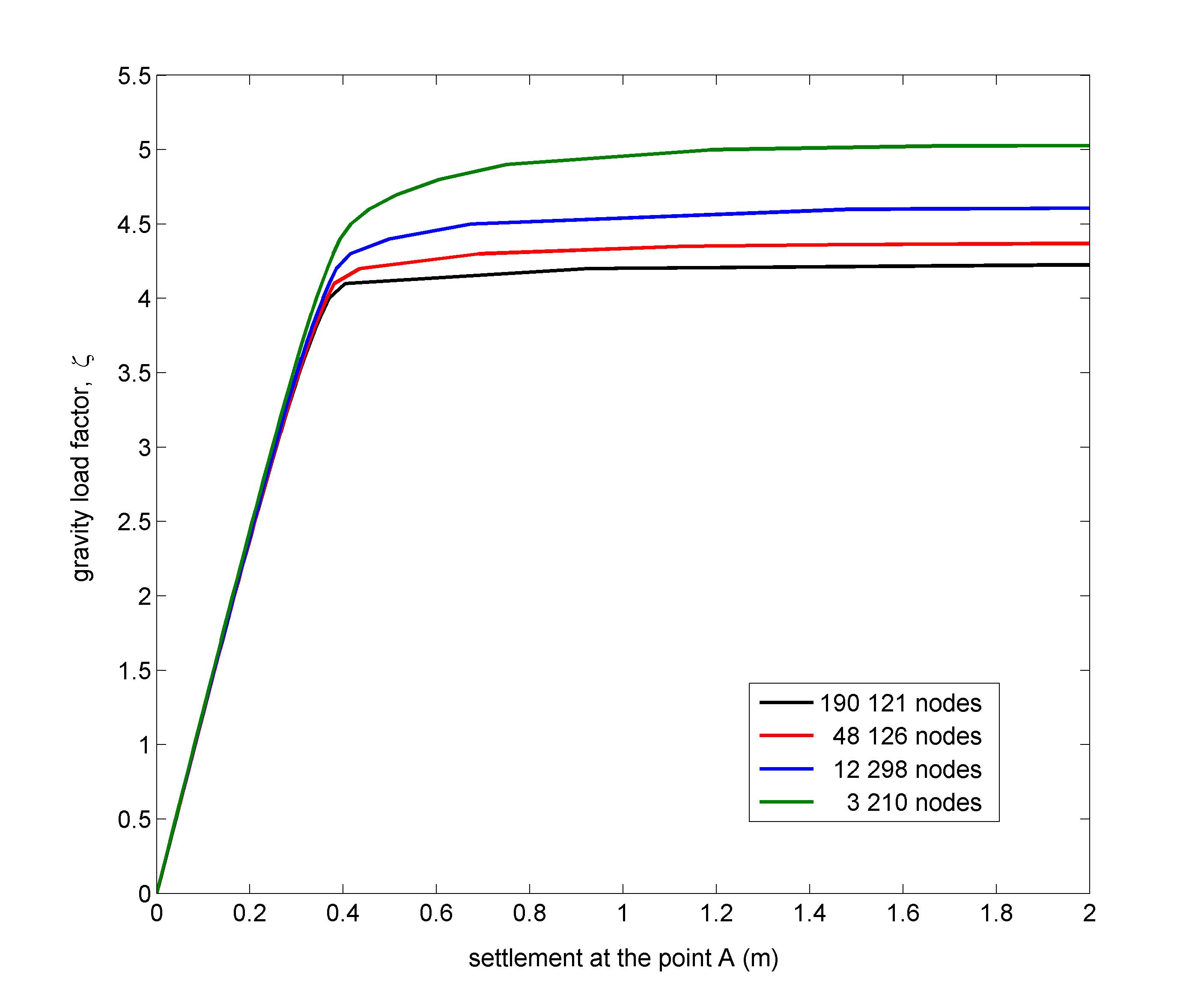}
   \caption{\small{Loading paths for the associative perfect plastic model and $P1$ elements.}}
   \label{fig.load_path_h_P1}
\end{minipage}
\hfill
\begin{minipage}[t]{0.47\textwidth}
  \center
  \includegraphics[width=\textwidth]{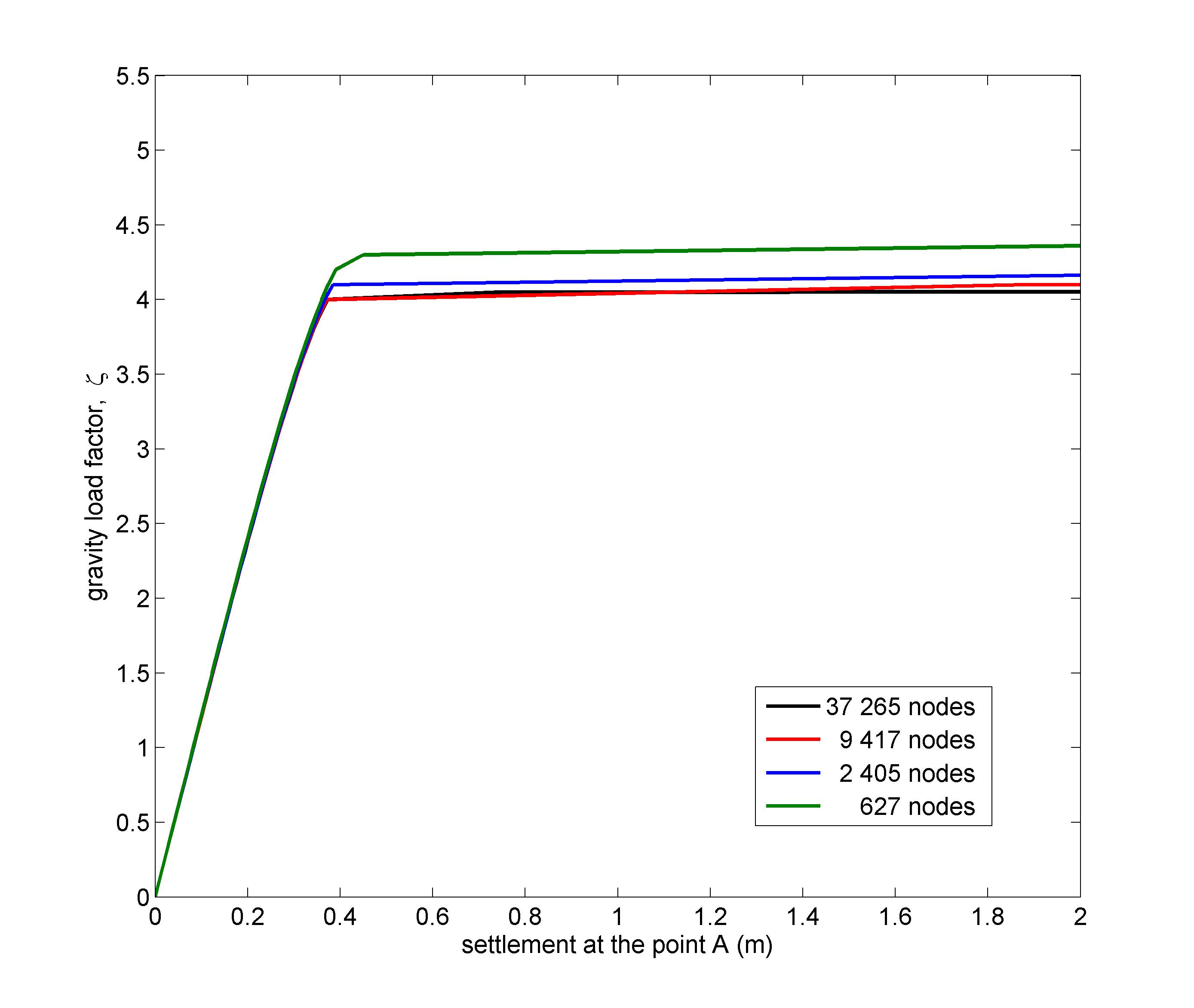}
   \caption{\small{Loading paths for the associative perfect plastic model and $Q2$ elements.}}
   \label{fig.load_path_h_Q2}
\end{minipage}
\end{figure}

For illustration, we add Figure \ref{fig.multiplier_perf_plas} and \ref{fig.displacement_perf_plas} with plastic multipliers and total displacements at collapse, respectively. The figures  are in accordance with literature. 

\begin{figure}[htbp]
\begin{minipage}[t]{0.47\textwidth}
  \center
  \includegraphics[width=\textwidth]{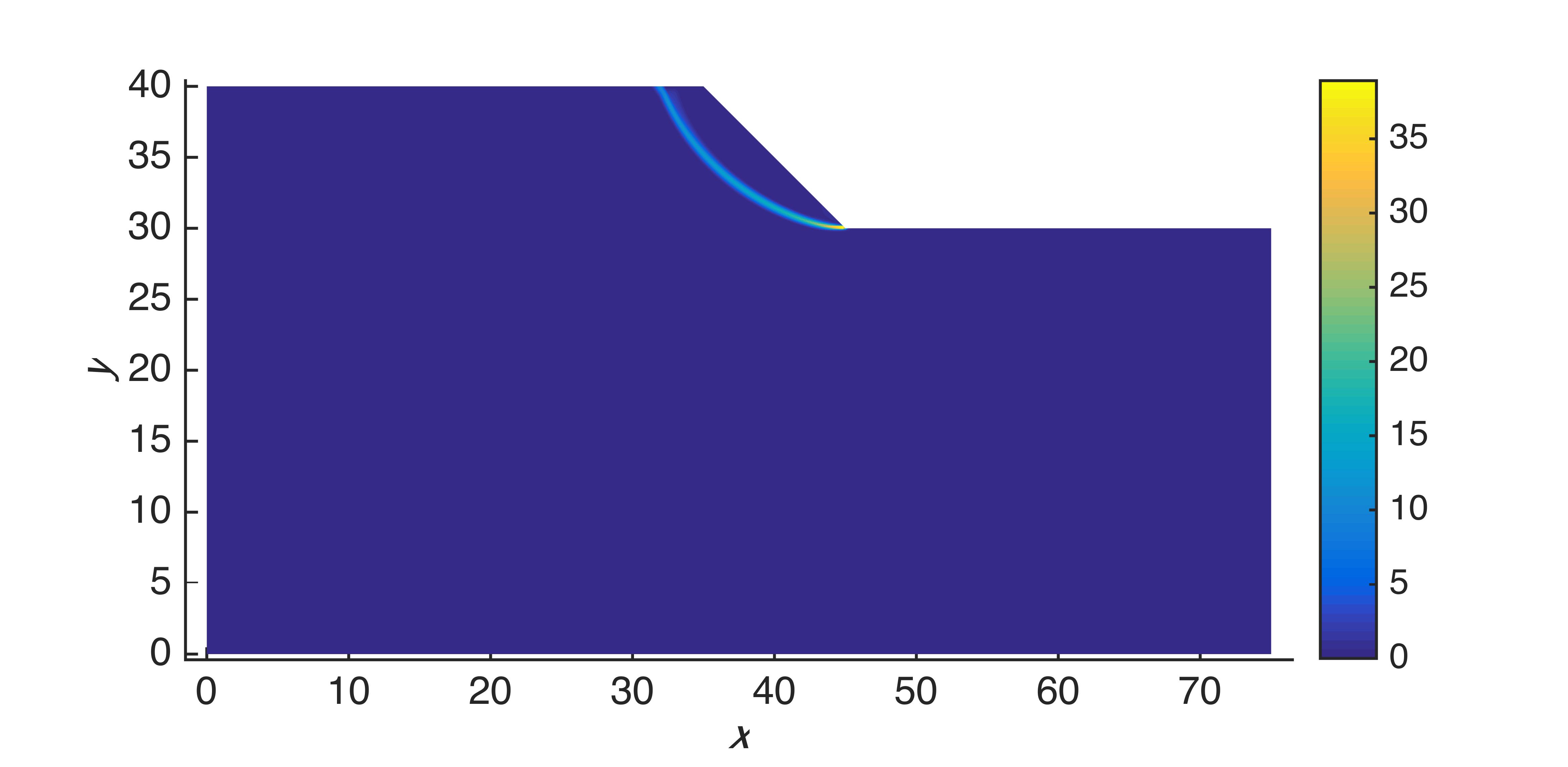}
   \caption{\small{Plastic multipliers at collapse for the finest $Q2$-mesh.}}
   \label{fig.multiplier_perf_plas}
\end{minipage}
\hfill
\begin{minipage}[t]{0.47\textwidth}
  \center
   \includegraphics[width=\textwidth]{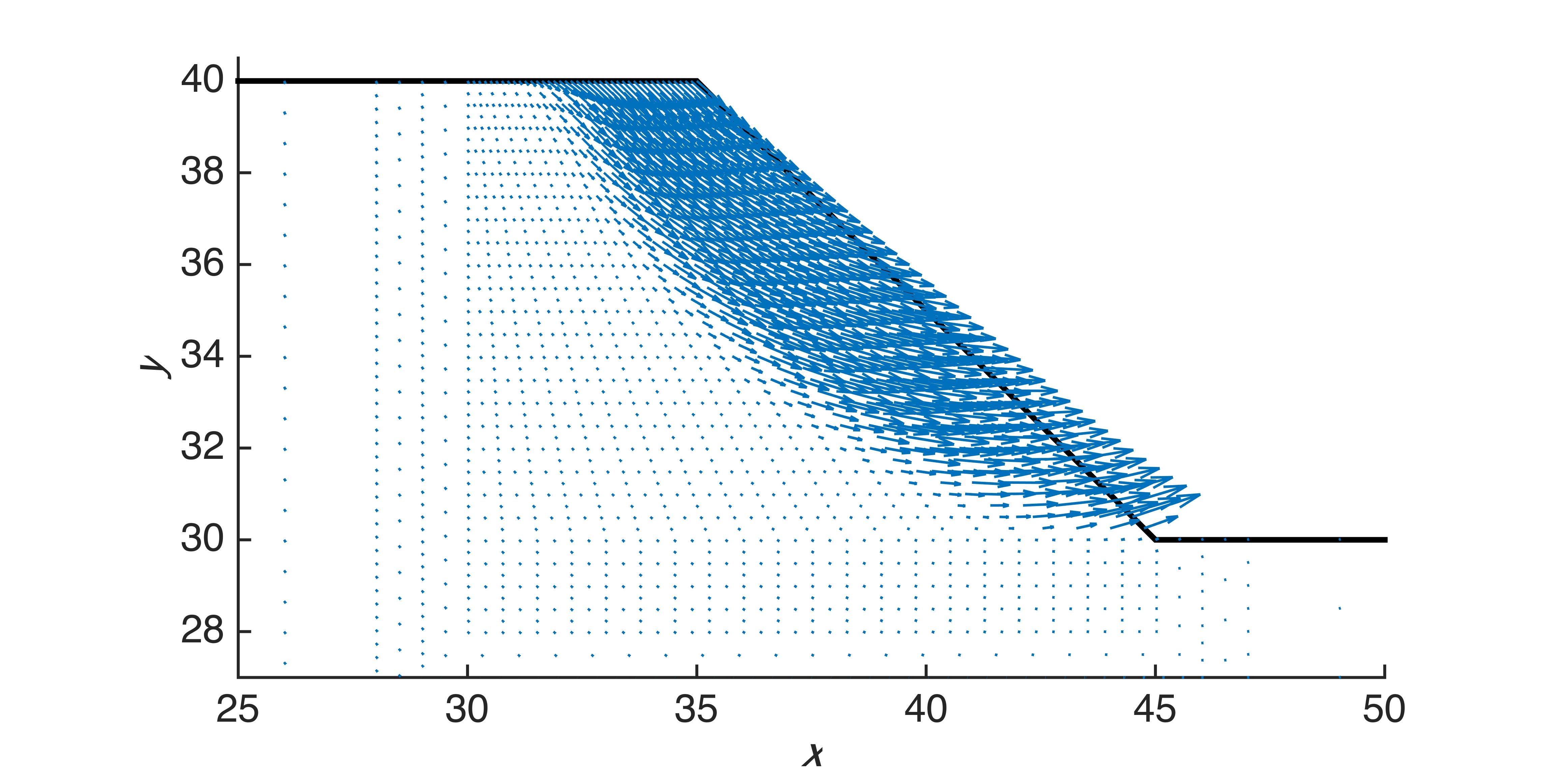}
   \caption{\small{Displacements at collapse (detail) for the $Q2$-mesh with 2405 nodes.}}
   \label{fig.displacement_perf_plas}
\end{minipage}
\end{figure}

To compare the current return-mapping scheme with improved one, we have also considered the nonassociative model with nonlinear hardening where 
$$\psi=10^{\circ},\; c_0=40\, \mbox{kPa}, \; \tilde H=10000\,\mbox{kPa},\; H(\bar\varepsilon^p)=\min\left\{c-c_0,\;\tilde H\bar\varepsilon^p-\frac{\tilde H^2}{4(c-c_0)}(\bar\varepsilon^p)^2\right\}.$$
Here, $\tilde H$ represents the initial slope of the hardening function and the material response is perfect plastic for sufficiently large values of the hardening variable.
This nonassociative model yields a slightly lower values of the limit load factors and also the other results are very similar to the associative model. The related graphical outputs are available in \cite[SS-DP-NH]{Mcode}. Further, in vicinity of the limit load, we have observed lower rounding errors for the improved return-mapping scheme and thus lower number of Newton steps is necessary to receive the prescribe tolerance than for the current scheme. However, the computational time for both schemes are practically the same since return to the apex happens only on a few elements lying in vicinity of the yield surface.




\subsection{The simplified Jirasek-Grassl model}

To be the simplified Jirasek-Grassl (JG) model applicable for the investigated soil material we fit its parameters using the associative perfect plastic Drucker-Prager (DP) model as follows:
$e=1$, $\bar f_c =3c\xi/(\sqrt{3}-\eta)$, $\bar f_t=0$, $B_g=1000$, $s=5$, $A_g=s\eta$ and $m_0=\sqrt{3}s-6$. 
Recall that $e=1$ implies $\varrho_e=\varrho$. Further the value of $\bar f_c$ corresponds to the uniaxial compressive strength computed from the Drucker-Prager model. To eliminate the 
influence of the exponential term in the function $m_g$, the value of $B_g$ is chosen sufficiently large. Then the model is insensitive on $\bar f_t$ and one can vanish it. Finally, we require the same flow direction for both the models under the uniaxial compressive strength.
Since the yield function in the JG model is normalized in comparison to the DP model it is convenient to introduce the following relation between the plastic multipliers: $\triangle\lambda_{DP}=\frac{s}{\bar f_c}\triangle\lambda_{JG}$, where $s$ is a scale factor. Then the values of $m_0$ and $A_g$ are determined from the following equations:
$$s\eta=m_g'(\bar f_c/3)\approx A_g,\quad\frac{s}{\sqrt{2}}=\sqrt{6}+\frac{m_0}{\sqrt{6}}.$$
To be $m_0$ positive, $s$ must be greater than $2\sqrt{3}$. To be in accordance with results of the DP model, we set $s=4.9$.\footnote{For smaller values of $s$, the limit load factor is underestimated and for greater values of $s$ the limit load factor is overestimated.} Comparison of yield surfaces (in the meridean plane) and flow directions for the DP and JG models  is illustrated in Figure \ref{fig_model_comparison}. Here, the fixed value $\triangle\lambda_{DP}=0.001$ is used for vectors representing the flow directions.
\begin{figure}[htbp]
\begin{center}
   \includegraphics[width=0.8\textwidth]{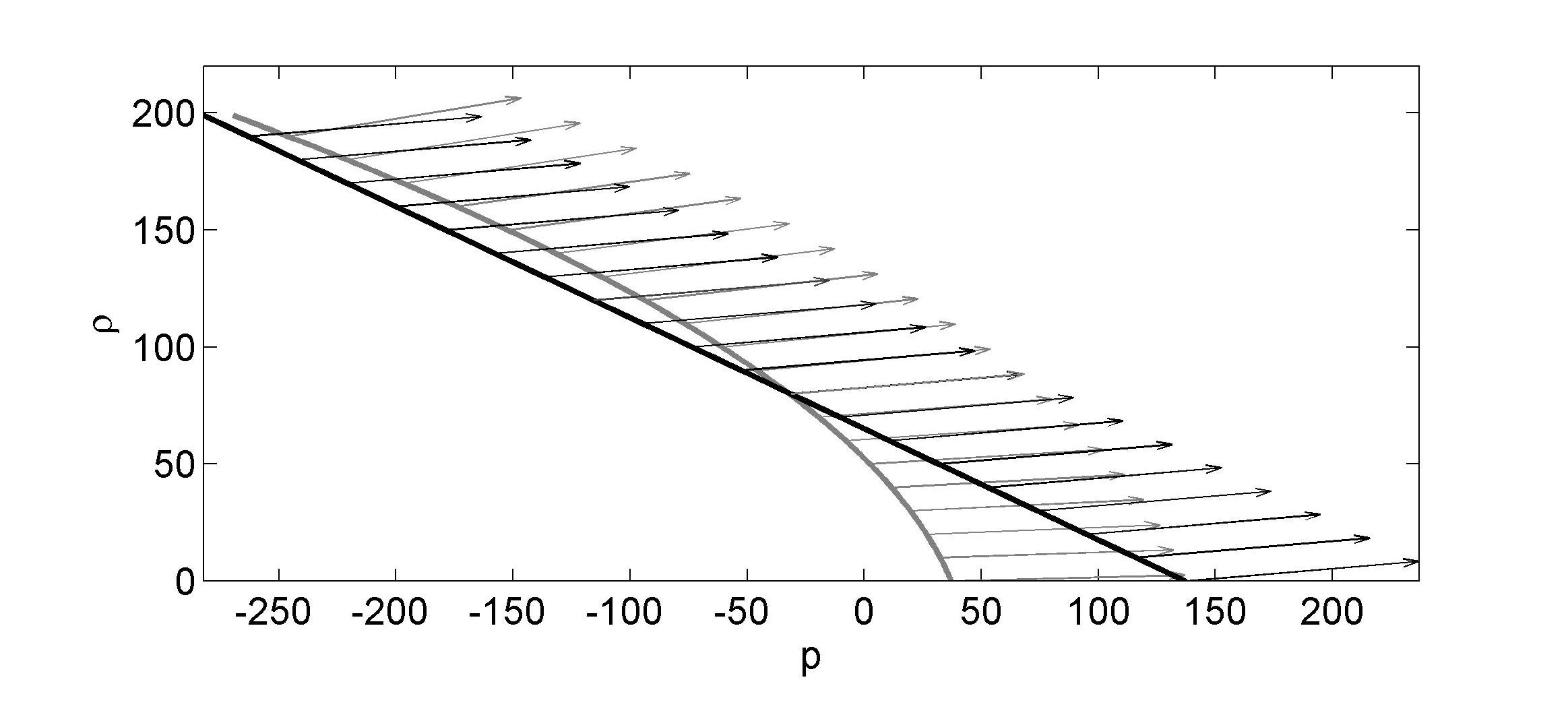}
   \caption{\small{Comparison of the yield surfaces (in the meridean plane) and the flow directions for the DP model (black) and the JG model (grey).}}
   \label{fig_model_comparison}
\end{center}
\end{figure}

Loading curves for the investigated $P1$, $Q2$ meshes and the JG model are depicted in Figure \ref{fig.load_path_h_JG} and \ref{fig.load_path_h_Q2_JG}. We observe much faster convergence of the $P1$-loading curves than for the DP model. Moreover, the results for $P1$ and $Q2$ elements are comparable. The computed values of the limit load factor on the finest $P1$ and $Q2$ meshes are 4.124, and 4.107, respectively.
\begin{figure}[htbp]
\begin{minipage}[t]{0.47\textwidth}
  \center
   \includegraphics[width=\textwidth]{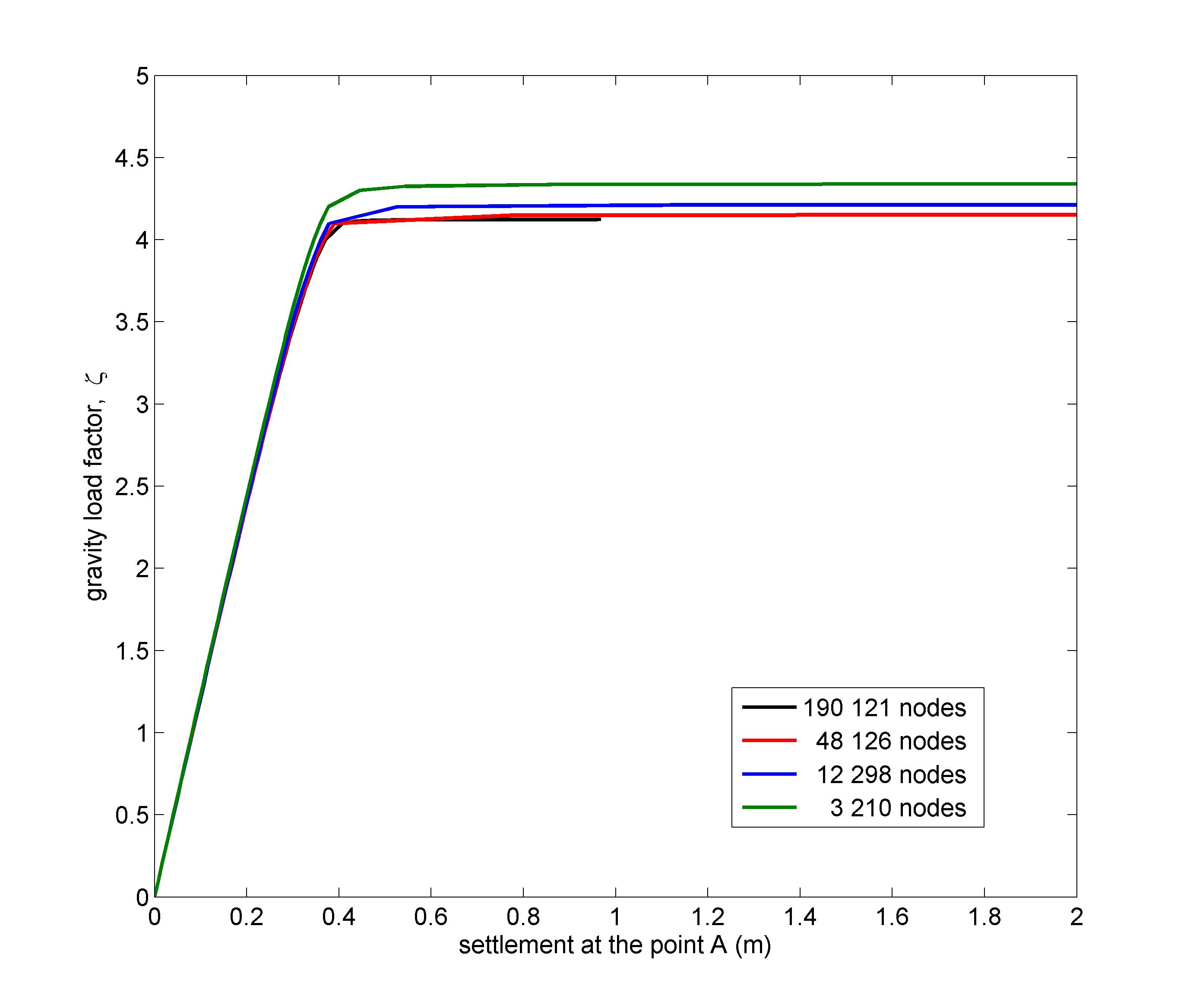}
   \caption{\small{$P1$ - loading paths for the simplified Jirasek-Grassl model.}}
   \label{fig.load_path_h_JG}
\end{minipage}
\hfill
\begin{minipage}[t]{0.47\textwidth}
  \center
  \includegraphics[width=\textwidth]{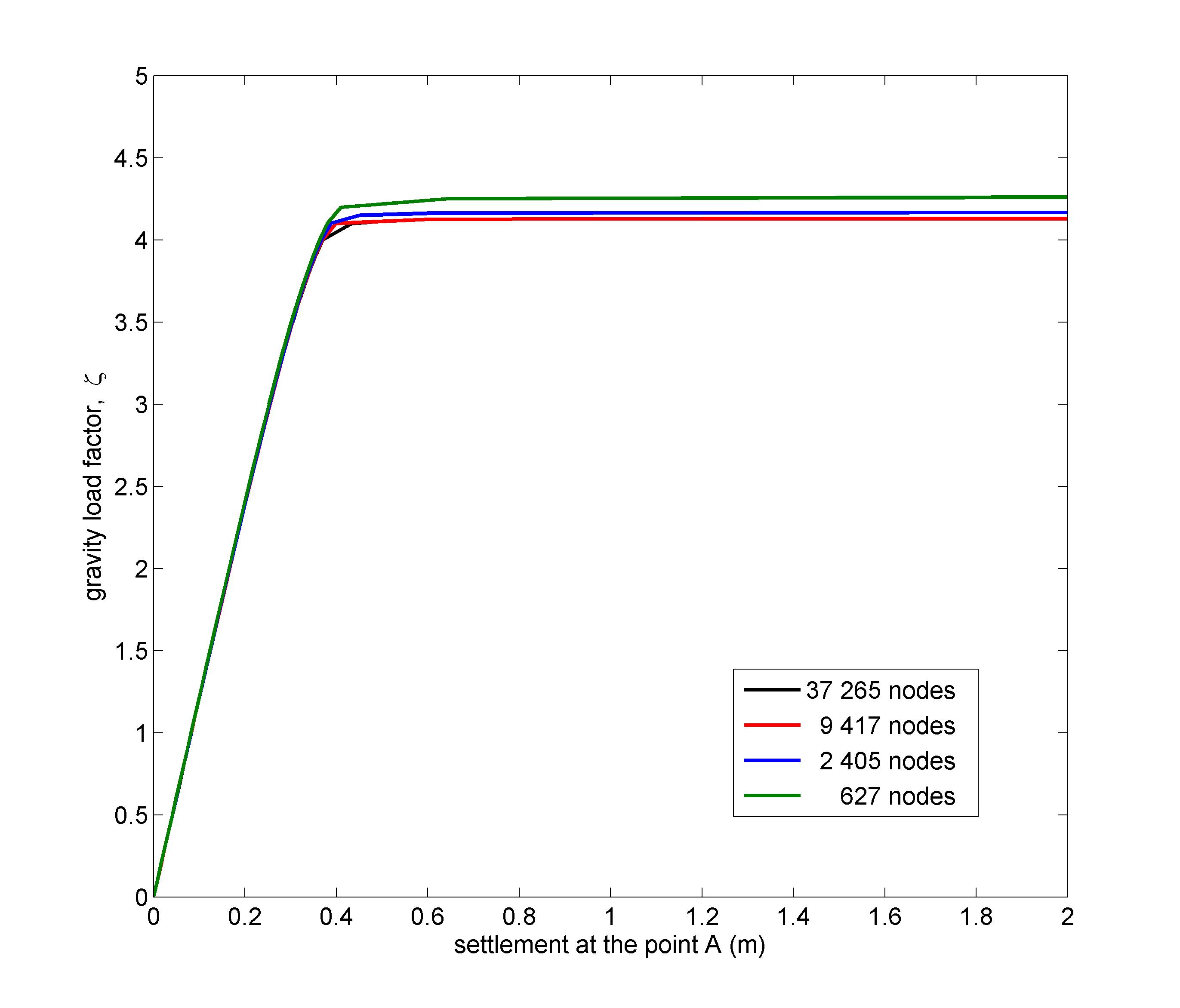}
   \caption{\small{$Q2$ - loading paths for the simplified Jirasek-Grassl model.}}
   \label{fig.load_path_h_Q2_JG}
\end{minipage}
\end{figure}

\section{Conclusion}

The main idea of this paper is that the subdifferential formulation of the plastic flow rule is also useful for computational purposes and numerical analysis. Namely, it has been shown that such an approach improves the implicit return-mapping scheme for non-smooth plastic pseudo-potentials as follows.
\begin{itemize}
\item The unique system of nonlinear equations is solved regardless on a type of the return.
\item It can be a priori determined the type of the return from a given trial state for some models (without knowledge of the solution).
\item The scheme can be more correct than the current one, and its form enables to study properties of constitutive operators like existence, uniqueness and semismoothness.
\end{itemize}
In this paper (PART I), the new technique has been systematically built on a specific class of models containing singularities only along the hydrostatic axis. Beside an abstract model, two particular models have been studied: The Drucker-Prager and the simplified Jirasek-Grassl model. However, the presented idea seems to be more universal. For example, it has been successfully used for the Mohr-Coulomb model in "PART II"  \cite{CKKSZ15b}.

\section*{Acknowledgements}
The authors would like to thank to Pavel Mar\v s\'alek for generating the quadrilateral meshes with midpoints.
This work has been supported by the project 13-18652S (GA CR) and the European Regional Development Fund in the IT4Innovations Centre of Excellence project (CZ.1.05/1.1.00/02.0070). 


\end{document}